\documentclass[pra,aps,superscriptaddress,twocolumn,nopacs,nofootinbib]{revtex4}
\usepackage{graphicx,color,epstopdf}
\usepackage{amsmath,amsfonts,enumerate,amsthm,amssymb,mathtools}
\usepackage{enumitem}
\usepackage{thmtools,thm-restate}
\usepackage{bbold}
\usepackage{hyperref}
\usepackage{multirow}
\usepackage{subfigure} 
\usepackage{mathdots} 

\usepackage{eufrak}
\usepackage{array}
\usepackage{blkarray}

\hypersetup{
colorlinks=true, 
linkcolor=blue, 
citecolor=blue 
}
\usepackage{amsmath}
\usepackage{amsfonts}
\usepackage{amssymb}
\usepackage{graphicx}
\usepackage{color}
\usepackage[latin1]{inputenc}
\usepackage{wrapfig}

\usepackage{amsthm}

\newtheorem{thm}{Theorem}

\newtheorem{lem}[thm]{Lemma}

\newtheorem{defn}{Definition}
\newtheorem{rmk}[thm]{Remark}

\newcommand{\be}{\begin{equation}}
\newcommand{\ee}{\end{equation}}

\DeclareMathOperator{\Tr}{Tr}

\renewcommand{\v}[1]{\ensuremath{\boldsymbol #1}}

\newsavebox{\smlmat}
\savebox{\smlmat}{$\left(\begin{smallmatrix}
1&3&4\\
2&0&0\\
5&0&0\\
\end{smallmatrix}\right)$}

\newsavebox{\smmat}
\savebox{\smmat}{$\left(\begin{smallmatrix}
0&0&2\\
0&7&0\\
5&0&1\\
\end{smallmatrix}\right)$}

\newsavebox{\mat}
\savebox{\mat}{$\left(\begin{smallmatrix}
4&0&0&0\\
5&3&6&7&\\
0&0&2&0\\
0&0&0&1\\
\end{smallmatrix}\right)$}

\newsavebox{\matP}
\savebox{\matP}{$\left(\begin{smallmatrix}
1&1&0&0\\
1&0&1&0&\\
1&0&0&1\\
1&0&0&0\\
\end{smallmatrix}\right)$}

\begin{document}
\title{Thermal Processes and State Achievability}

\author{Pawe\l{} Mazurek}
\affiliation{Institute of Theoretical Physics and Astrophysics, National Quantum Information Centre, Faculty of Mathematics, Physics and Informatics, University of Gda\'nsk, 80-308 Gda\'nsk, Poland}

\begin{abstract}
A complete characterization of the set of states that can be achieved through Thermal Processes (TP) is given by describing all vertices, edges and facets of the allowed set of states in the language of thermomajorization curves. TPs are linked to transportation matrices, which leads to the existance of extremal TPs that are not required in implemenation of any transition allowed by TPs, for every dimension $d\geq 4$ of the state space. A property of the associated graphs, biplanarity, which differentiates between these extremal TPs and the necessary ones, is identified.
\end{abstract}

\maketitle

\section{Introduction}

One of the approaches to a quantum description of microscale systems interacting with a macroscale environment is through the so-called Resource Theories, with Thermal Operations (TO) \cite{Janzing00, Streater95} being one of the most fruitful. Recently, it allowed for a full characterization of possible total energy-conserving transitions between states diagonal in the eigenbasis of a local Hamiltonian \cite{Janzing00, Ruch76, Horodecki13}, description of deterministic work extraction and work cost of a state formation \cite{Aberg2013,Horodecki13} and lead to the formulation of microscale II laws of thermodynamics in the form of a family of equations restricting evolution of R\'enyi entropies \cite{Brandao2015}. Yet, many questions still remain open. Most remarkably, the quantum thermodynamic description of systems with coherences within TO needs perfecting, as only bounds on the evolution of coherences exist for a general qudit case \cite{Cwiklinski2015}. 

More generally, the physical applicability of bounds delivered by TO has been questioned, as full control over the system is assumed within TO description. For diagonal initial states, transitions allowed by TO are described by Thermal Processes (TP), which have been shown to be implementable as a sequence of operations involving partial thermalizations (which require only weak coupling between the system and the bath \cite{Alicki07}) and manipulation of the system Hamiltonian, though at cost of using an ancilla \cite{Perry16}. As this sequence can be highly non-trivial, it is tempting to ask whether every transition allowed by a TP for a state in $d$ dimensional space can be performed as a convex combination of sequences of TPs acting on lower-dimensional spaces of the system (we will denote them by TP(n), with $n<d$, where $d$ being dimension of the Hilbert space of the system), as they not only do not imply the necessity of manipulating the energy gap of the Hamiltonian, but also are guaranteed to be of limited length, due to a convex structure of a set of states $\rho_{init}^{TP(d)}$ available through TPs from a given initial state $\rho_{init}$. 

Such studies have been carried out by authors of \cite{Lostaglio16b, Mazurek2017}, with a conclusion that a set of states achievable by TP(2), denoted by $\rho_{init}^{TP(2)}$, is strictly smaller from $\rho_{init}^{TP(d)}$ for $d>2$ and some $\rho_{init}$. Furthermore, there always exists a state $\rho_{init}$ for which the same relation holds between $\rho_{init}^{TP(d-1)}$ and $\rho_{init}^{TP(d)}$ for arbitrary $d\geq3$ \cite{Mazurek2017}. Even in the restricted scenario of TP(2) transformations, length of sequences of TP(2) required for all allowed transitions is not known for general $d$, only an upper bound exists \cite{Lostaglio16b}. Finally, one may ask if all transformations TP(d) are required to reach an arbitrary state within $\rho_{init}^{TP(d)}$, for all $\rho_{init}$. Trivially, by allowing mixing between operations, we can focus on extremal TPs. Then one can still ask: is the experimenter required to be able to perform every extremal TP in order to implement all the allowed transitions in the state space? A precise description of the structure of two convex sets: TP(d) of Thermal Processes and $\rho_{init}^{TP(d)}$ of states achievable by them, allows us to answer this question negatively.

\begin{figure}[h]
\centering
\includegraphics[width = 1\linewidth]{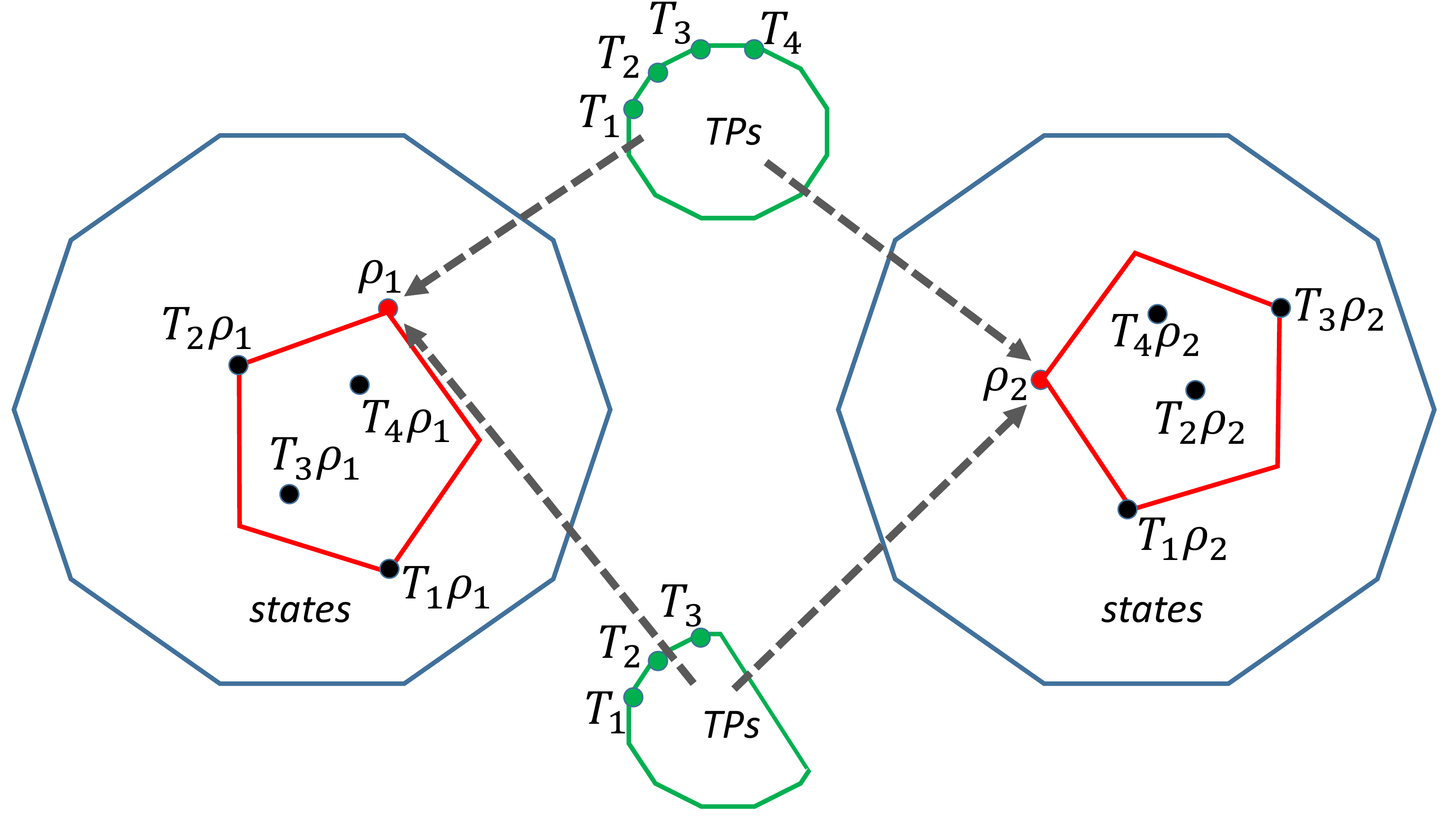}
\caption{\label{nonbi} 
Sets of diagonal states (blue polygons) and Thermal Processes (upper green polygon) for a system in a space of dimension $d$ (not directly manifested in the picture). For different initial states $\rho_{1}$ and $\rho_{2}$, different sets of states achievable by TPs can be obtained: $\rho_{1}^{TP}$ and $\rho_{2}^{TP}$, respectively (red polygons). Some extremal points of TPs, $T_{2}$ and $T_{3}$, map to an extremal point or to the interior of $\rho^{TP}$, depending on the initial state $\rho$. Non-biplanar extremal TPs, like $T_{4}$, always map into the interior of a set $\rho^{TP}$, no matter which initial state $\rho$ from the blue polygone is selected. When taking convex compositions of TPs is allowed, these extremal TPs can be discarded with no harm to attainable set of states (lower green polygon). Non-biplanar extremal TPs exist for all $d\geq 4$.    
}
\end{figure}

In this paper we provide a complete characterization of extremal points, edges and hyperfaces of $\rho_{init}^{TP(d)}$ in terms of the so called thermomajorization curves. Furthermore, based on the connection between TPs and transportation matrices \cite{Jurkat1967} we introduce the notion of biplanar transportation matrices and define biplanar extremal TPs. We show that every such TP is uniquely connected with a family of pairs of states, and therefore it is possible to determine every biplanar extremal TP(d) with a thermomajorization diagram. Moreover, we show that, for a generic $\rho_{init}$, all extremal TPs that do not belong to the class of biplanar extremal TPs map into the interior of $\rho_{init}^{TP(d)}$, therefore they lack a clear physical significance. We show a construction of such a process, proving the existence of such processes for arbitrary finite, non-zero temperature, for all $d\geq 4$.

We will proceed to the main part of the paper after introducing definitions and basic properties of Thermal Processes and transportation matrices.

\section{Preliminaria}

\subsection{Thermal Processes}

We start with a brief description of Resource Theory of Thermal Operations. A $d$-dimensional system in a state $\rho\in\mathcal{B}(\mathcal{H_{S}})$, $\mathcal{H_{S}}\cong\mathbb{C}^d$, is associated with a Hamiltonian $H_{S}=\sum_{i=0}^{d-1}E_{i}|E_{i}\rangle\langle E_{i}|$. For a bath $\mathcal{B}(\mathcal{H_{B}})$ with a Hamiltonian $H_{B}$ we define a Gibbs state $\mathcal{B}(\mathcal{H_{B}})\ni\rho_{\beta}^{B}=\exp^{-\beta H_{B}}/\Tr[\exp^{-\beta H_{B}}]$, where $\beta=\frac{1}{kT}$, and $k$ is Boltzmann constant, while $T$ is temperature. 

For a given $\beta$, we define a set of Thermal Operations (TO) as all maps $\mathcal{E}:\mathcal{B}(\mathcal{H_{S}})\rightarrow\mathcal{B}(\mathcal{H_{S}})$ that can be constructed by the following operations:   
\begin{itemize}
\item one can perform an arbitrary unitary $U$ on $\mathcal{B}(\mathcal{H_{S}}\otimes\mathcal{B}(\mathcal{H_{B}}))$ that conserves the total energy: $[U,H_{S}+H_{B}]=0$.
\item one can extend the system by adding an arbitrary ancilla $\mathcal{B}(\mathcal{H_{A}})$, $\mathcal{H_{A}}\cong\mathbb{C}^{d'}$, with a Hamiltonian $H_{A}$, in a Gibbs state $\rho_{A}=\exp^{-\beta H_{A}}/\Tr[\exp^{-\beta H_{A}}]$.
\item one can remove an arbitrary ancilla $\mathcal{B}(\mathcal{H_{A'}})$, $\mathcal{H_{A'}}\cong\mathbb{C}^{d''}$, with a Hamiltonian $H_{A'}$, in a Gibbs state $\rho_{A'}=\exp^{-\beta H_{A'}}/\Tr[\exp^{-\beta H_{A'}}]$.
\end{itemize} 

This leads to a set of trace preserving, completely-positive maps on a system $\rho\rightarrow\mathcal{E}(\rho)=\Tr_{B}\Big[U [\rho\otimes \rho_{\beta}^{B}] U^{\dagger}\Big]$.

Under an assumption that $H_{S}$ has a non-degenerated spectrum, a TO acting on a state that is diagonal in the basis of its Hamiltonian ($[\rho_{S},H_{S}]=0$) cannot lead to creation of coherences in this basis. This follows from the fact that evolution of non-diagonal elements of the density matrix in independent from the diagonal \cite{Lostaglio15}. Moreover, it is easy to see that TO conserve Gibbs state of the system, $\rho_{\beta}^{S}$. Therefore, action $\mathcal{E}(\rho)$ of a TO $\mathcal{E}$ on a state $\rho$ diagonal in $H_{S}$, represented by a vector $\v{p}$ comprising of its eigenvalues, $\rho=diag[\v{p}]$, can be associated with an action $T\v{p}$ of a left-stochastic matrix $T$ preserving a Gibbs vector. We call this matrix a Thermal Process (TP):

\begin{defn}
A set of Thermal Process $TP(d)_{\beta,H_{S}}$ is a set of $d\times d$ matrices $T$ satisfying $\v{1}^{T}T=\v{1}^{T}$ and $T\v{g}=\v{g}$, where $\v{1}^{T}=[1,\dots,1]$ and $\v{g}$: $g_{i}=q_{i,0}/\sum\limits_{j} q_{j,0}$ are vectors of lenght $d$. Here, $q_{m,n}=e^{-\beta(E_{m}-E_{n})}$, with $E_{i}$ being eigenvalues of $H_{S}$.  
\end{defn}

In what follows, we will be skipping indexes $d,\beta,H_{S}$ when it does not lead to confusion, assume that $E_{0}=0$, and assume all states $\rho$ satisfy $[\rho,H_{S}]$ for a non-degenerated Hamiltonian $H_{S}$, and therefore represent them by vectors $\v{p}$ carrying information about their occupations on energy levels of $H_{S}$. 

On the other hand, for every TP there exists a TO on bath and a system diagonal with respect to $H_{S}$, that performs this TP on a system. \cite{Horodecki13}. Therefore,  all transformations allowed for diagonal states within Thermal Operations Resource Theory can be equivalently characterized by TPs. 

\begin{rmk}
Denote by $\rho_{init}^{TP}$ a set of states that can be obtained through Thermal Processes $TP$  from a state $\rho_{init}$. This set is a convex polytope. 
\end{rmk}

\begin{proof}
From the definition of a Thermal Process we see that $TP$ is a convex polytope. Now, take $\rho_{1}, \rho_{2}\in \rho_{init}^{TP}:$ $T_{1}\rho_{init}=\rho_{1}$, $T_{2}\rho_{init}=\rho_{2}$, where $T_{1}, T_{2}\in TP$. 
For every $0\leq\alpha\leq 1$, a state $\alpha \rho_{1}+(1-\alpha)\rho_{2}$ belongs to $\rho_{init}^{TP}$ due to convexity of $TP$: $\alpha \rho_{1}+(1-\alpha)\rho_{2}=\alpha T_{1}\rho_{init}+(1-\alpha)T_{2}\rho_{init}=T_{3}\rho_{init}$, where $T_{3}=(\alpha T_{1}+(1-\alpha)T_{2})\in TP$. Therefore, $\rho_{init}^{TP}$ is convex. Each of $\rho_{out}\in\rho_{init}^{TP}$ can be represented as $\rho_{out}=\sum \alpha_{i}T^{ext}_{i}\rho_{init}$, where $0\leq\alpha_{i}\leq1$ and $\sum_{i}\alpha_{i}=1$ and $\{T^{ext}_{i}\}_{i}$ is a set of extremal points of $TP$. If $\rho_{out}$ is extremal, then $T^{ext}_{k}\rho_{init}=T^{ext}_{l}\rho_{init}$ for all pairs of $T^{ext}_{i}$ contributing with a non-zero coefficient to the decomposition of $\rho_{out}$. Therefore, number of extremal points of $\rho_{init}^{TP}$ cannot be bigger than number of extremal points of $TP$, and as the latter set is a convex polytope, the former one is a polytope as well.
\end{proof}

The set of states achievable by TPs from a given initial state $\rho_{init}$ is fully characterized by a criterion exploiting representation of a vector $\v{p}$ on a the so-called thermomajorization diagram (see Fig. \ref{termow}).

\begin{defn}[Thermo-majorization curve]

Define a vector $\v{s}=(q_{00},q_{10},q_{20},\dots,q_{d-1,0})$. For every state $\rho_{S}$ commuting with $H_{S}$, let a vector $\v{p}$ represents occupations $p_{i}$ of energy levels $E_{i}$, $i=0,1,\dots,d-1$. Choose a permutation $\pi$ on \v{p} and \v{s}, such that it leads to a non-increasing order of elements in a vector $\v{d}$, $d_{k}=\big(\frac{\sum_{i=0}^{k}(\pi \v{p})_i}{ \sum_{i=0}^{k} (\pi \v{s})_i }\big)$, $k=0,\dots,d-1$. A set of points $\{\sum_{i=0}^{k}(\pi\v{p})_i,\sum_{i=0}^{k}(\pi\v{s})_i\}_{k=0}^{d-1}\cup\{0,0\}$, connected by straight lines, defines a curve associated with the state $\rho$. We denote it by $\beta(\rho)$ and call a thermomajorization curve of state $\rho$ represented by $\v{p}$. 
\end{defn}

\begin{figure}[h]
\centering
\includegraphics[width = 1\linewidth]{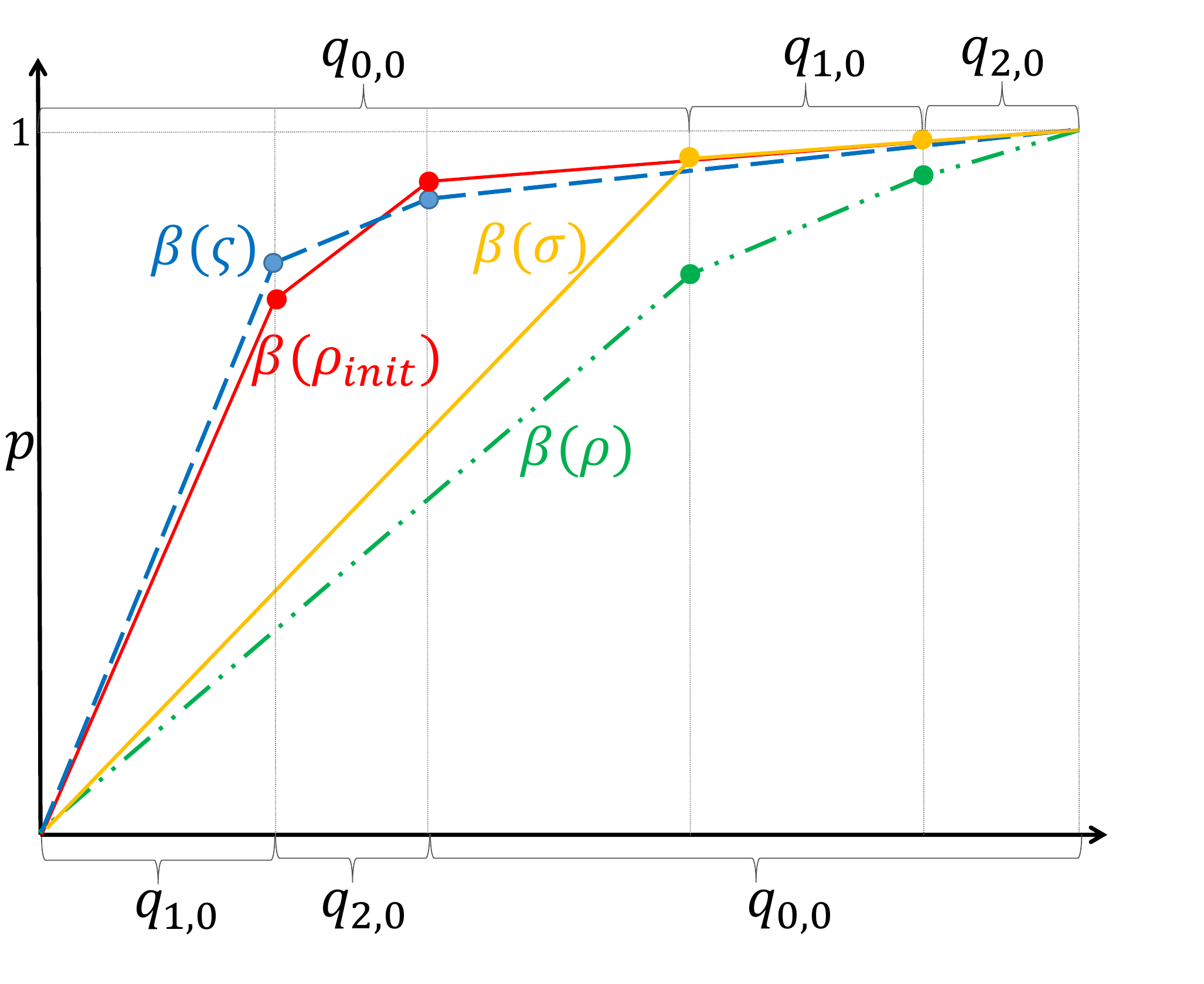}
\caption{\label{termow} 
Thermomajorization diagram for a d=3 system, and certain $H_{S}$ and $\beta$ defining $q_{00}=1$, $q_{10}$ and $q_{20}$. According to Lemma \ref{lem5}, $\rho,\sigma\in\rho_{init}^{TO(3)_{\beta,H_{S}}}$, but $\zeta\notin \rho_{init}^{TO(3)_{\beta,H_{S}}}$ and $\rho_{init}\notin \zeta^{TO(3)_{\beta,H_{S}}}$. Elbows of curves are indicated by circles. Curves $\beta(\sigma)$ and $\beta(\rho)$ are \textit{thermomajorized} by curve $\beta(\rho_{init})$. Curve $\beta(\sigma)$ is \textit{tightly thermomajorized} by curve $\beta(\rho_{init})$. States $\rho_{init}$ and $\zeta$ have $\beta$-order $(2,3,1)$, while states $\sigma$, $\rho$ have $\beta$-order $(1,2,3)$.}
\end{figure}

Points $\{\sum_{i=0}^{k}(\pi\v{p})_i,\sum_{i=0}^{k}(\pi\v{s})_i\}_{k=0}^{d-1}$ will be called elbows of a curve $\beta(\v{p})$. 
The curve is concave due to a non-increasing order of elements in $\v{d}$. Let us note that there might be more than one permutation leading to a creation of a concave curve $\beta(\rho)$. The vector $\pi(1,\dots,d)^{T}$ will be called a $\beta$-order of $\rho$. It shows modification of the order of segments that had to be done in order to assure convexity of $\beta(\rho)$.

\begin{lem} [\cite{Horodecki13}]
A transition from $\rho_{init}\in\mathcal{B}(\mathcal{H_{S}})$, $\mathcal{H_{S}}\cong\mathbb{C}^d$ to $\rho_{out}\in\mathcal{B}(\mathcal{H_{S}})$ under $TPs$ is possible if and only if $\beta(\rho_{init})$ thermomajorizes $\beta(\rho_{out})$, i.e. all elbows of $\beta(\rho_{out})$ lie on $\beta(\rho_{init})$ or below it.
\label{lem5}
\end{lem}

For the sake of characterization of a set of extremal points of $\rho_{init}^{TP}$ (Sec. \ref{sec:states}), we single out a specific relation between two curves:

\begin{defn}[Tight thermomajorization]
If a curve $\beta(\rho)$ has all elbows on a curve $\beta(\sigma)$,  $\beta(\sigma)$ tightly thermo-majorizes $\beta(\rho)$.
\end{defn}

\subsection{Transportation matrices}
Below we introduce a notion of transportation matrices, which properties will be useful in characterizing the connection between sets of states achievable by thermal processes, $\rho_{init}^{TP}$, and a set $TP$ alone.

\begin{defn}[Transportation matrix]
A transportation matrix $M$ is a $m\times n$ matrix with non-negative entries determined by two vectors $\v{c}$ and $\v{r}$ of lengths $m$ and $n$, respectively, in a way that all entries from the $i$-th row(column) of M sum to $r_{i}$ ($c_{i}$), and $\sum_{i}c_{i}=\sum_{j}r_{j}=C$.  
\end{defn}

For all pairs of non-negative vectors $\v{c}$ and $\v{r}$ satisfying the summation condition, an associated transportation matrix always exists: if $C=0$, it is a matrix with all entries equal $0$; for other cases we can construct $M$ as $M_{i,j}=r_{i}c_{j}/C$.  

For every pair of vectors the set of transportation matrices is a convex polytope, with convexity stemming from linearity of matrix addition, and limited number of extremal points coming from the requirement that for $M$ to be extremal it cannot have more than $m+n-1$ non-zero elements -- there is a a limited number of selection of locations for these elements within a matrix, and for every valid choice, values of these elements are uniquely determined. 

A set of extremal points of a transportation polytope is fully characterized by the following constructive algorithm [\cite{Jurkat1967}, Theorem 4.1]:

\begin{thm} [Extremal Points of a Transportration Polytope]\label{Th3}
A transportation matrix with defining vectors $\v{c}$, $\v{r}$, $\sum c_{i}=\sum r_{i}=C$ is extremal if and only if it can be constructed by repeating the following step, starting with a matrix with no values assigned:
\begin{itemize}
\item Pick a position $(i,j)$ in the matrix that has no assigned value, and fill it with $\min(r_{i},c_{j})$. If $r_{i}\leq c_{j}$ ($r_{i}\geq c_{j}$), fill all remaining entrances within an $i$-th row ($j$-th column) of the matrix with $0.$ ( This implies that if $r_{i}=c_{j}$, all the remaining entrances within $i$-th row and $j$-th column will be filled with 0). Update the values $r_{i}\rightarrow r_{i}-\min(r_{i},c_{j})$, $c_{j}\rightarrow c_{j}-\min(r_{i},c_{j})$. Updated vectors are non-negative and satisfy a summation criterion, therefore they define a transportation matrix. 
\end{itemize}
\end{thm} 

The procedure ends with $\v{r}=\v{0}$ and $\v{c}=\v{0}$, when we assign $0$ to all entries without any values assigned at that moment. All entries of the matrix have been determined, with at most $n+m-1$ positive ones. 

We define an important class of extremal transportation matrices in the following way:

\begin{defn} [Biplanar Extremal Transportration Matrix]\label{Def4}
An extremal transportation matrix with defining vectors $\v{c}$, $\v{r}$, $\sum c_{i}=\sum r_{i}=C$ is biplanar if it can be constructed by the following procedure:
\begin{itemize}
\item Pick a position $(i,j)$ in the matrix and fill it with $\min(r_{i},c_{j})$. If $r_{i}\leq c_{j}$ ($r_{i}\geq c_{j}$), fill all remaining entrances within an $i$-th row ($j$-th column) of the matrix with $0.$ ( This implies that if $r_{i}=c_{j}$, all the remaining entrances within $i$-th row and $j$-th column will be filled with 0). Update the values $r_{i}\rightarrow r_{i}-\min(r_{i},c_{j})$, $c_{j}\rightarrow c_{j}-\min(r_{i},c_{j})$. Updated vectors are non-negative and satisfy a summation criterion, therefore they define a transportation matrix. 
\item If, for indexes $i$ and $j$ from the last section, $r_{i}\neq 0$ ($c_{j}\neq 0$), choose a position without assigned value from the $i$-th row ($j$-th column) as the starting position from step 1, and apply it. If both $r_{i}=c_{j}=0$, then pick another position in the matrix that has no assigned value, $(i',j')$, and return to step 1 by substituting $i'\rightarrow i$, $j'\rightarrow j$. 
\item If $\v{r}=\v{0}$ and $\v{c}=\v{0}$, abort. All entries of the matrix have been determined, with at most n+m-1 positive ones.
\end{itemize}
\end{defn} 

The name of the class comes from the property of the graphs associated with adjacency matrices of these matrices. To every extremal matrix we can assign its adjacency matrix. Adjacency matrix is a matrix of the same size as the transportation matrix, with $0$ entries, except where its corresponding transportation matrices has a positive entry, them the adjacency matrix entry is $1$. If we associate rows of an adjacency matrix with graph vertices on the right side of the bipartite graph, and columns of an adjacency matrix with graph vertices on the left side of the bipartite graph, and connect vertices by edges whenever the corresponding entry of the adjacency matrix is 1, then all vertices of these graph are assigned to disjoint subgraphs in a way that, within a selected subgraph, every vertex is connected to any other vertex by a single path on edges (see Fig. \ref{forest}).

\begin{figure}[h]
\centering
\includegraphics[width = 1\linewidth]{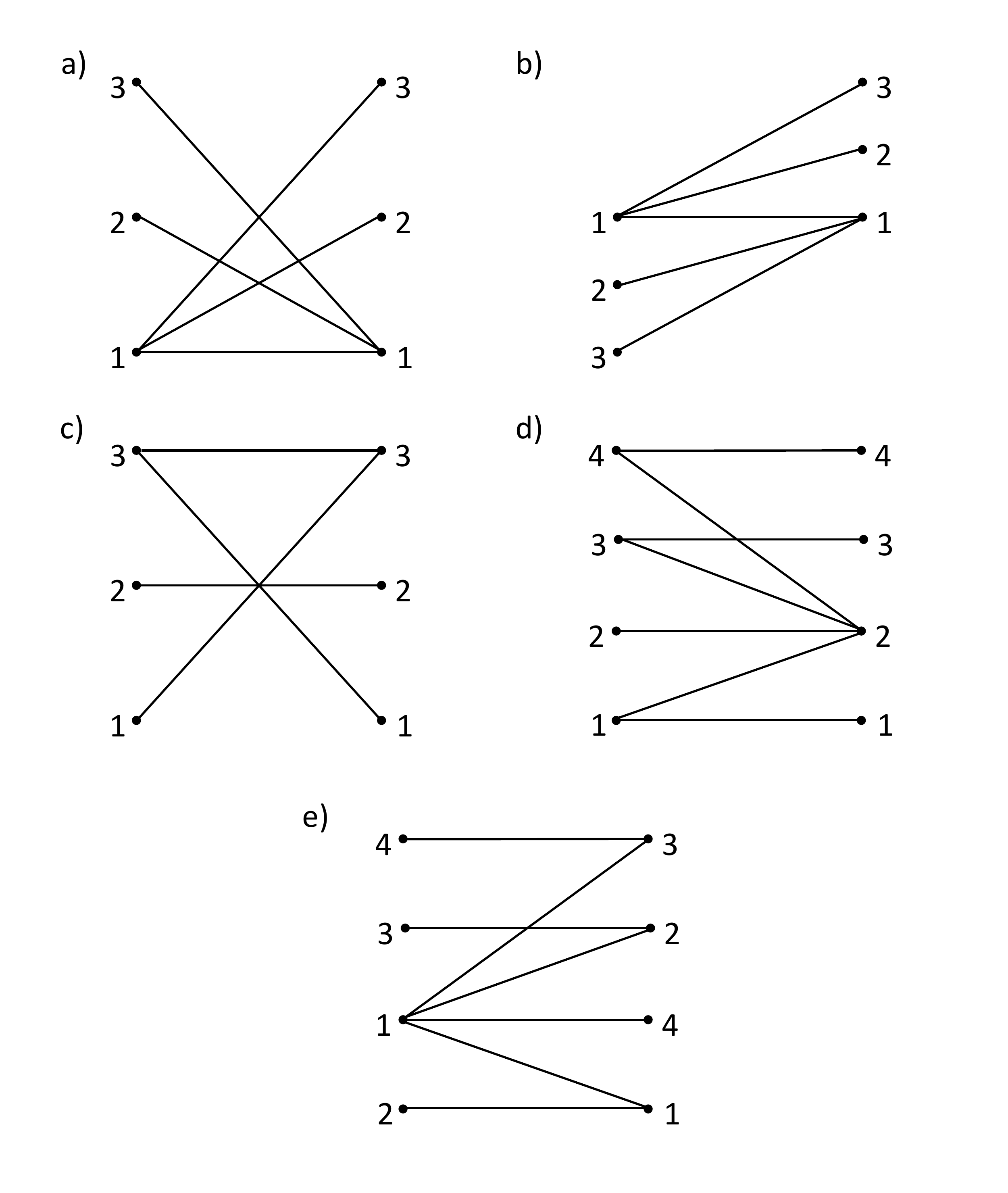}
\caption{\label{forest} 
Forests with no isolated vertices on bipartite graphs for chosen extremal transportation matrices: a)~\usebox{\smlmat}, with $\v{r}=(8,2,5)$, $\v{c}=(8,3,4)$; b) the same as in a); c) ~\usebox{\smmat}, with $\v{r}=(2,7,6)$, $\v{c}=(5,7,3)$; d) ~\usebox{\mat}, with $\v{r}=(4,21,2,1)$, $\v{c}=(9,3,8,8)$; (e) ~\usebox{\matP}, with $\v{r}=(2,2,2,1)$, $\v{c}=(4,1,1,1)$.  Notice that a graph in b) is a plain version of a biplanar graph in a). Graphs in d) and e) cannot be driven into plain forms by isomorphisms that do not switch vertices between sides of the bipartite structure. Graph in c) can be driven into such form, and, as opposed to the rest of the graphs, is a forest composed of 2 trees, instead of one tree. Labels of the vertices on the left(right) side of a graph correspond to columns (rows) of the respective extremal transportation matrix.}
\end{figure}

In the graph language, it means that extremal points of transportation matrices can be associated with forests (sets of trees) on bipartite graphs with no isolated vertices \cite{Klee68}. This property is also visible directly from the the construction of Theorem \ref{Th3}, which enforces that every vertex of the bipartite graph is included in some subgraph if we take vectors $\v{r}$ and $\v{c}$ to be positive, and also sets this subgraphs to be trees (otherwise, the procedure would have to allow for the subgraphs to have cycles, which is forbidden by the fact that at every step we put $0$ elements at non-determined entries of a matrix along some column or a row, so they cannot be assigned a positive value in further steps, and cycles cannot be formed).  

Now, biplanar extremal matrices are characterized by a property that for their associated graphs one can perform an isometric transformation on the vertices that preserves a bipartite structure (i.e. one can change positions of the vertices withing each side of the bipatite graph, without changing nearest neighbors within the graph), such that none of the edges cross (the graph becomes plain). This property justifies the name -- while existance of an isometric form of a graph with no edges crossing is a defining sign of a graph being planar, here we restrict isomorphic maps to those which preserve a bipartite structure. This property characterizes matrices from the biplanar class, because picking an initial element in step 1 of Definition \ref {Def4} fixes initial column and row numbers, as well as defines an edge connecting corresponding vertices on an associated bipartite graph. If in step 2 we decide to continue the procedure on the same column, we would add an edge to a new vertex on the right side of the graph. In contrary, if we we were to define a new element in the same row as the previous one, it is equivalent to adding edges starting from a selected vertex on a right side of the graph. Selection of a new, independent element starts a construction of a new tree. 

What will be crucial in our applications, it implies that there exists an order of vertices (determined in the direction from the bottom to the top) of the bipartite graph such that the graph is plain. On the other hand, if such an order exists, then the extremal transportation matrix can be created by the procedure from Definition \ref{Def4}, thus the transportation matrix belongs to the class of biplanar extremal matrices. Note that in general it does not have to be the case, i.e. there exist bipartite graphs with forests of edges and no isolated vertices, such that they cannot be driven to a plain form by isomorphic transformations conserving a bipartite structure; yet these graphs can be associated with extremal transportation matrices (see Fig.\ref{forest}d)).

Note that a transportation matrix does not mix between elements belonging to different trees (Fig.\ref{forest}c)). For rectangular transportation matrices of size $n\times n$, maximal number of trees is achieved by a diagonal extremal transportation matrix (if it exists for given $\v{c}$ and $\v{r}$). Such a matrix has also the smallest possible number of positive elements (n).

\section{Geometry of the set of states achievable by TPs}\label{sec:states}
Let us denote by $Extr[A]$ a set of extremal points of a convex set $A$. The structure of the set of states achievable through TPs from a given initial state $\rho_{init}$ is described by the following Thereom:

\begin{thm}[Extremal points of states achievable from $\rho_{init}$ by Thermal Processes]
\label{Extremal}
A state $\rho$ belongs to $Extr[\rho_{init}^{TP}]$ if and only if $\beta(\rho)$ is tightly thermomajorized by $\beta(\rho_{init})$.
\end{thm}

\begin{proof}[Proof: if part]
If $\beta(\rho)$ is thermomajorized by $\beta(\rho_{init})$, then $\rho$ belongs to $\rho_{init}^{TP}$ due to Lemma \ref{lem5}. Moreover, $\rho$ is an extremal point of this set: if it was not true, then there would exist two different states $\sigma_{1}$, $\sigma_{2}$ belonging to $\rho_{init}^{TP}$, such that $\rho$ could be created as their non-trivial convex combination. But this would imply that thermomajorization curve of at least one of these states is not thermomajorized by $\beta(\rho_{init})$, therefore  contradicting the fact that $\sigma_{1}\in\rho_{init}^{TP}$. This implication is visible from the following reasoning: for a thermomajorozation curve $\beta(\rho)$, by $\beta_{i}(\rho)$ we will denote the slope of the segment of length given by $s_{i}$. $\rho=a\sigma_{1}+(1-a)\sigma_{2}$ implies $\beta_{i}(\rho)=a\beta_{i}(\sigma_{1})+(1-a)\beta_{i}(\sigma_{2})$ for every $i=0,\dots,d-1$. Therefore, if we choose $i$ such that $\beta_{i}(\rho)$ is the highest slope of $\beta(\rho)$, from the fact that $\beta(\rho)$ is \textit{tightly thermomajorized} by $\beta(\rho_{init})$ we see that $\beta_{i}(\rho)$ is the maximal slope that the segment of lenght $s_{i}$ can take, such that $\beta(\rho)$ is \textit{thermomajorized} by $\beta(\rho_{init})$. But from the convex combination relation we have that either $\beta_{i}(\sigma_{1})>\beta_{i}(\rho)$ or $\beta_{i}(\sigma_{2})>\beta_{i}(\rho)$ or $\beta_{i}(\sigma_{1})=\beta_{i}(\sigma_{2})=\beta_{i}(\rho)$. Therefore, in the first two cases $\sigma_{1}\notin \rho_{init}^{TP}$ or $\sigma_{2}\notin \rho_{init}^{TP}$ and we arrive with the thesis. In the third case, we proceed to the segment characterized by $i$ such that it has second-highest slope of $\beta(\rho)$. Again, as $\beta(\rho)$ is tightly thermomajorized by $\beta(\rho_{init})$, $\beta_{i}(\rho)$ is the highest possible slope of the segment of length $s_{i}$, provided the slope of the segment of the highest slope is fixed according to the previous step. Again, creating $\rho$ as a mixture of $\sigma_{1}$ and $\sigma_{2}$ would lead to a conclusion that either $\beta_{i}(\sigma_{1})>\beta_{i}(\rho)$ or $\beta_{i}(\sigma_{2})>\beta_{i}(\rho)$ or $\beta_{i}(\sigma_{1})=\beta_{i}(\sigma_{2})=\beta_{i}(\rho)$. By iterating this procedure for consecutive segments, according to descending order of the slopes of $\beta(\rho)$, we see that the only allowed decomposition of $\rho$ that is tightly thermomajorized by $\rho_{init}$ into $\sigma_{1,2}\in\rho_{init}^{TP}$ is for $\sigma_{1,2}=\rho$. Therefore, $\rho\in Extr[\rho_{init}^{TP}]$. 
\end{proof}

\begin{proof}[Proof: only if part]
We show that every state $\rho\in\rho_{init}^{TP}$ can be represented as  $\rho=\sum p_{i}\sigma_{i}$, where $\sum p_{i}=1$, $0<p_{i}$ and thermomajorization curves of $\sigma_{i}$, $\beta(\sigma_{i})$, are all tightly thermomajorized by $\beta(\rho_{init})$ -- therefore, this set of extremal points is complete. First, we notice that for a 2 level system, every thermal process can be described with just two extremal TPs: $\mathcal{I}=(\begin{smallmatrix}1&0\\0&1\end{smallmatrix})$ and $\mathcal{B}=(\begin{smallmatrix}1-q_{10} & 1\\ q_{10} & 0\end{smallmatrix})$. Therefore, if $\rho\in\rho_{init}^{TP(2)_{\beta,H_{S}}}$, then $\rho=((1-\alpha)\mathcal{I}+\alpha\mathcal{E})\rho_{init}$, for $0\leq\alpha\leq 1$. The corresponding curve, $\beta(\rho(\alpha))$, has an elbow on two possible vertical lines (see Fig.\ref{2levels}). For $\alpha=0$, the elbow is on the initial curve $\beta(\rho_{init})$, it goes down with increasing $\alpha$ to reach a line characterizing a Gibbs state $\beta(\rho_{\beta})$, switches lines and continous up, to reach $\beta(\rho_{init})$ again for $\alpha=1$. In this way, we can achieve all states of $\rho_{init}^{TP(2)_{\beta,H_{S}}}$ characterized by curves with elbows lying between $\beta(\rho_{init})$ and $\beta(\rho_{\beta})$ on two specified lines. In what follows, we will be using this to decompose a given $d$-level state belonging to $\rho_{init}^{TP(d)_{\beta,H_{S}}}$ into two states that have equal occupations on all (d-2) levels, apart from selected two. Difference in occupations on these two levels makes one of the elbows from each of corresponding $\beta$-curves to lye on a different position, as in Fig.\ref{2levels} for $\beta(\rho(0)))$ and $\beta(\rho(1))$, while rest of the elbows from two curves lye on the same positions.   

\begin{figure}[h]
\centering
\includegraphics[width = 1\linewidth]{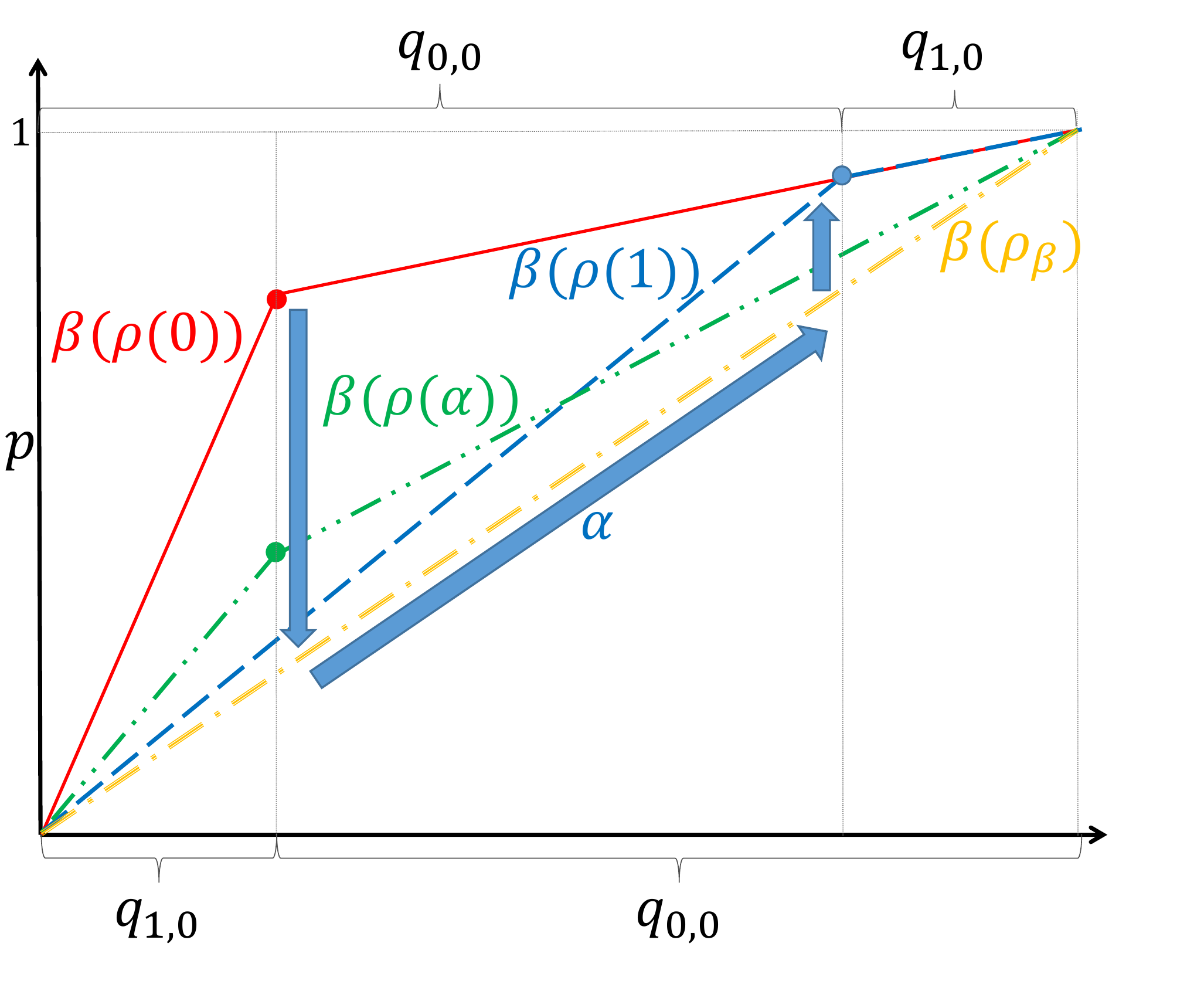}
\caption{\label{2levels} 
All states $\rho\in\rho_{init}^{TP(2)_{\beta,H_{S}}}$ can be represented as $\rho(\alpha)=((1-\alpha)\mathcal{I}+\alpha\mathcal{E})\rho_{init}$, with value $\alpha\in[0,1]$ determining the position of the only elbow of thermomajorization curve $\beta(\rho(\alpha))$. The movement of the elbow associated with continuous increase of $\alpha$ is marked by arrows. For $\alpha=1/(1+q_{10})$ marking a transition to a Gibbs state $\rho_{\beta}$, the elbow disappears, to reemerge for higher $\alpha$ on a different vertical line.}
\end{figure}

First, decompose the curve $\beta(\rho)=c_{1,1}\gamma_{1,1}'+c_{1,2}\gamma_{1,2}'$  (see Fig. \ref{step1}a) into a convex combination of curves $\gamma_{1,1}'$ and $\gamma_{1,2}'$ such that $\gamma_{1,1}'$ has the same $\beta$-order as $\beta(\rho)$, i.e. $\pi(\gamma_{1,1}')=\pi(\beta(\rho))$, and $\gamma_{1,2}'$ has the same $\beta$-order as $\beta(\rho)$, except for the last two segments, which are permuted: $\pi(\gamma_{1,2}')=P_{d-1,d}\pi(\beta(\rho))$, where $P_{d-1,d}$ marks the permutation between the segments that are at the position $d-1$ and $d$ in the following vector. Demand also that last elbows of $\gamma_{1,1}'$ and $\gamma_{1,2}'$ lie on $\beta(\rho_{init})$. Other elbows of $\gamma_{1,1}'$ and $\gamma_{1,2}'$ have the same positions as in $\beta(\rho)$. This decomposition corresponds to a decomposition of a state $\rho$ into two states that differ by occupations only on two selected levels. As we see from the case of two level systems, these requirements fix parameters $c_{1,1}$ and $c_{1,2}$, while preserving $c_{1,1}+c_{1,2}=1$. At the end, whenever curves $\gamma_{1,1}'$ or $\gamma_{1,2}'$ are not concave, we change order of segments such that we obtain proper thermomajorization curves, $\gamma_{1,1}$ and $\gamma_{1,2}$, respectively (Fig. \ref{step1}b). Therefore, we have obtained $\rho=c_{1,1}\rho(\gamma_{1,1})+c_{1,2}\rho(\gamma_{1,2})$, where $\rho(\gamma)$ marks a state $\rho$ associated with thermomajorization curve $\gamma$, and $\rho(\gamma_{1,1}), \rho(\gamma_{1,2})\in\rho_{init}^{TP}$.

\begin{figure}[h]
\centering
\includegraphics[width = 1\linewidth]{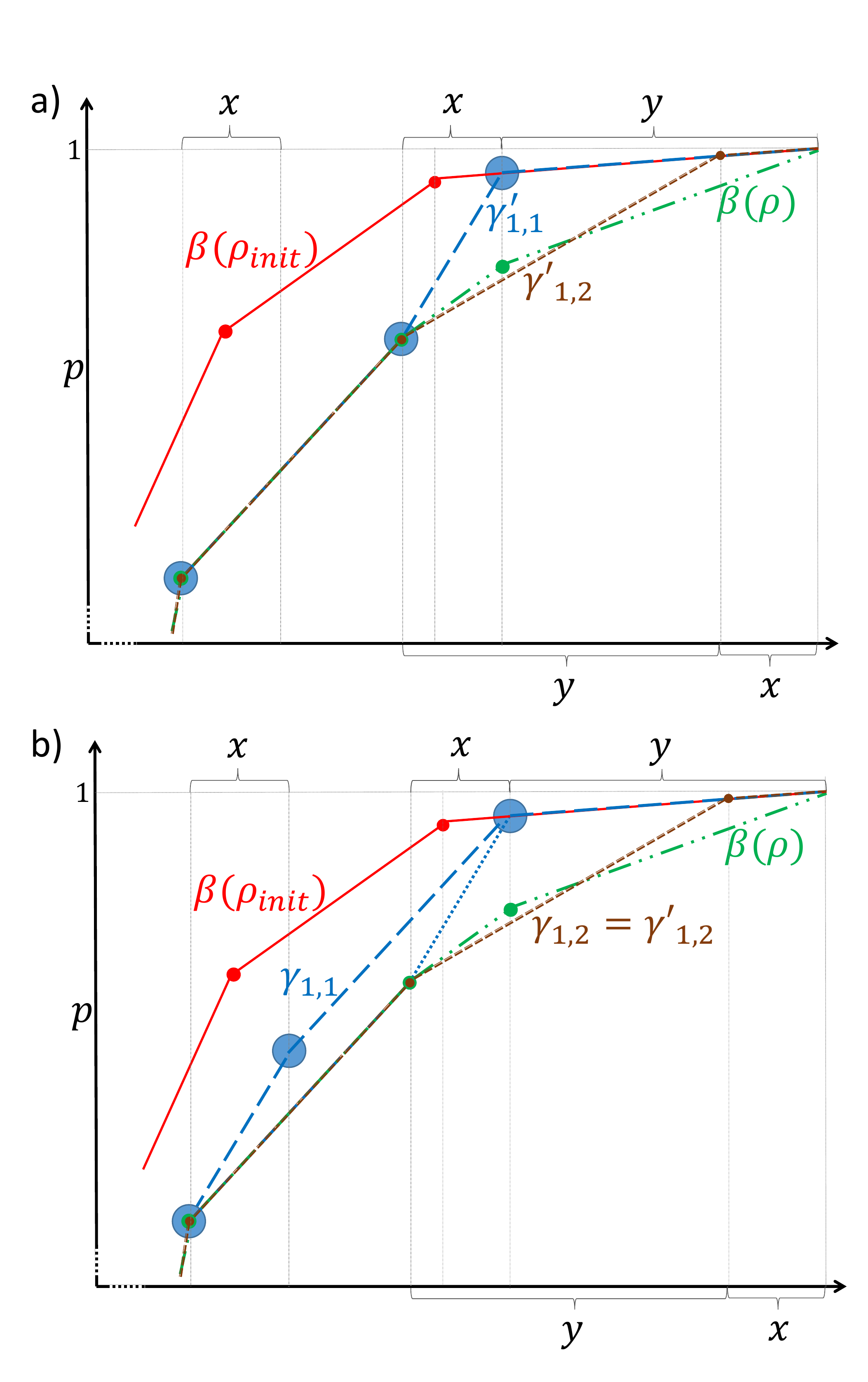}
\caption{\label{step1} 
Representing a state as a convex decomposition of extremal points (step 1). Decomposition of a state $\rho=c_{1,1}\rho(\gamma_{1,1})+c_{1,2}\rho(\gamma_{1,2})$, with last elbows of $\gamma_{1,1}$ and $\gamma_{1,2}$ lying on $\beta(\rho_{init})$. a) Validity of a construction comes from a decomposition of states of two-level systems (Fig. \ref{2levels}), trivially extended to states of higher dimension and equal occupations on the added levels. b) Permuting the segments turns $\gamma_{1,1}'$ and $\gamma_{1,2}'$ into $\gamma_{1,1}$ and $\gamma_{1,2}$, respectively, and asserts that the curves are concave.} 
\end{figure}

\begin{figure}[h]
\centering
\includegraphics[width = 1\linewidth]{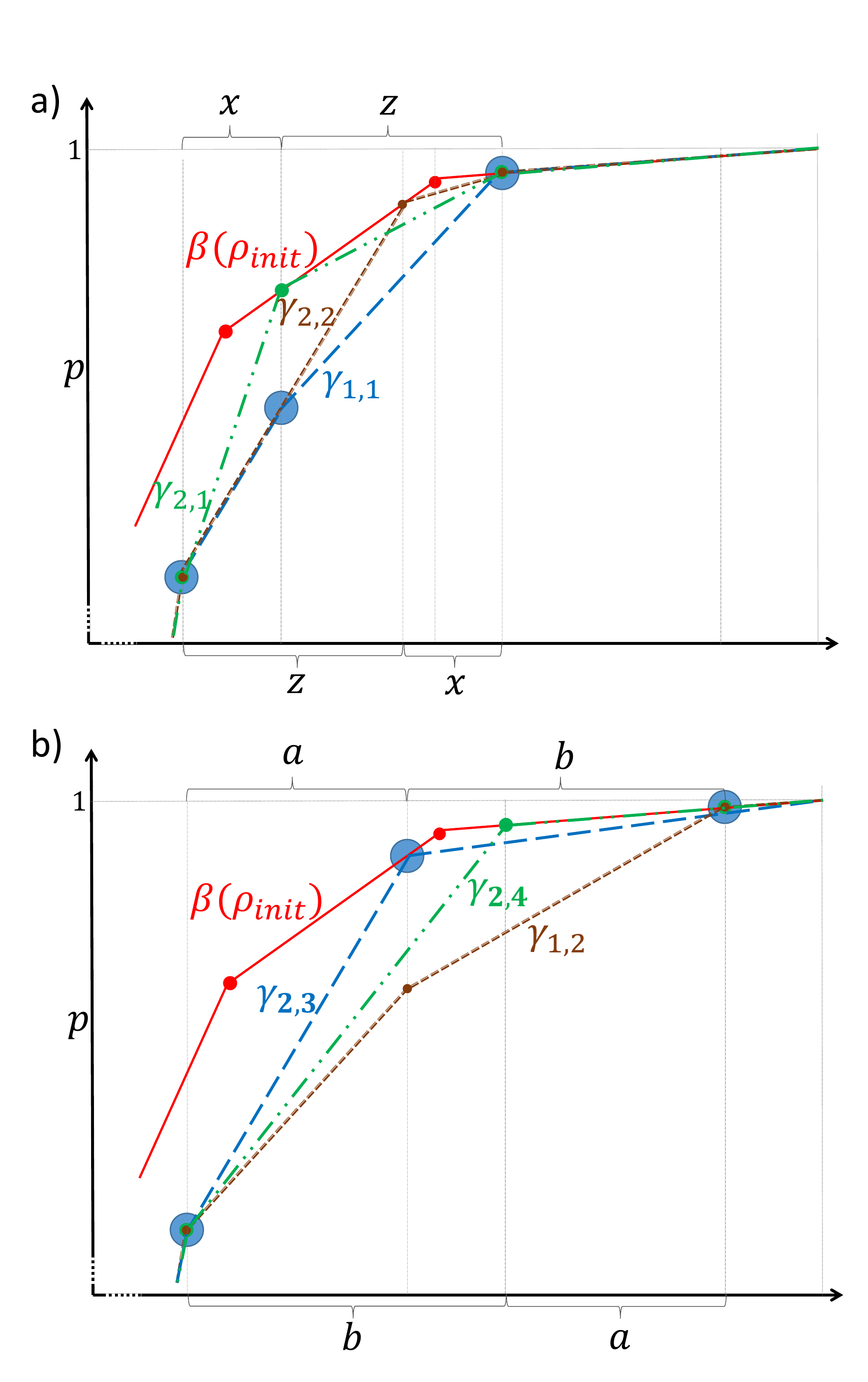}
\caption{\label{step2} 
Representing a state as a convex decomposition of extremal points (step 2). Decomposition of (a) $\rho(\gamma_{1,1})$ and (b) $\rho(\gamma_{1,2})$ into states represented by thermomajorization curves $\gamma_{2,i}$, $i=1,2,3,4$, with last 2 elbows lying on $\beta(\rho_{init})$.}
\end{figure}

In the second step, we decompose the curves $\gamma_{1,1}$ (Fig. \ref{step2}a) and $\gamma_{1,2}$ (Fig. \ref{step2}b) as $c_{1,1}\gamma_{1,1}=c_{2,1}\gamma_{2,1}'+c_{2,2}\gamma_{2,2}'$ and $c_{1,2}\gamma_{1,2}=c_{2,3}\gamma_{2,3}'+c_{2,4}\gamma_{2,4}'$, such that $\pi(\gamma_{2,1}')=\pi(\gamma_{1,1})$, $\pi(\gamma_{2,2}')=P_{d-2,d-1}\pi(\gamma_{1,1})$, $\pi(\gamma_{2,3}')=\pi(\gamma_{1,2})$, $\pi(\gamma_{2,4}')=P_{d-2,d-1}\pi(\gamma_{1,2})$, and such that two last elbows of $\gamma_{2,1}'$, $\gamma_{2,2}'$,   $\gamma_{2,3}'$, $\gamma_{2,4}'$ lie on $\beta(\rho_{init})$, while the position of the remaining elbows is like in the original lines, $\gamma_{1,1}$ and $\gamma_{1,2}$, respectively. Therefore, again we were using a decomposition of a given state into states with different occupations on just two energy levels, just now these two levels correspond to segments on a thermomajorization curves shifted towards left. Again, if necessary we permute segments to obtain concave curves $\gamma_{2,1}$, $\gamma_{2,2}$, $\gamma_{2,3}$ and $\gamma_{2,4}$,. In this way, parameters $c_{2,1}$, $c_{2,2}$, $c_{2,3}$, $c_{2,4}$, $\sum\limits_{i=1,4}c_{2,i}=1$ are fixed, and a decomposition $\rho=\sum\limits_{i=1,4}c_{2,i}\rho(\gamma_{2,i}$) is obtained. 

In general, we iterate this procedure for steps $j=1,\dots,d-1$, in each step dividing curves $c_{j-1,i}\gamma_{j-1,i}=c_{j,2i-1}\gamma_{j,2i-1}'+c_{j,2i}\gamma_{j,2i}'$ for all $i=1,\dots,2^{j}$ such that $\pi(\gamma_{j,2i-1}')=\pi(\gamma_{j-1,i})$, $\pi(\gamma_{j,2i}')=P_{d-j,d-j+1}\pi(\gamma_{j-1,i})$, such that last $j$ elbows of $\gamma_{j,2i-1}'$ and $\gamma_{j,2i}'$ lie on $\beta(\rho_{init})$, and the remaining elbows lie on $\gamma_{j-1,i}$. We take $c_{0,1}=1$ and $\gamma_{0,1}=\beta(\rho)$. We permute segments to obtain concave curves $\gamma_{j,2i-1}$ and $\gamma_{j,2i}$.  As each step fixes one more elbow of curves to lie on $\beta(\rho_{init})$, after $j=d-1$ steps all curves in the decomposition  $\beta(\rho)=\sum\limits_{i=1,2^{d-1}}c_{d-1,i}\gamma_{d-1,i}$ are tightly thermomajorozed by $\beta(\rho_{init})$. Also, as at each step we have $c_{j-1,i}=c_{j,2i-1}+c_{j,2i}$, it implies $\sum\limits_{i=1,2^{d-1}}c_{d-1,i}=1$. Therefore, after rewriting $c_{d-1,i}=p_{i}$ and $\rho(\gamma_{d-1,i})=\sigma_{i}$, we arrive with the convex decomposition $\rho=\sum p_{i}\sigma_{i}$. 
\end{proof}

\begin{rmk}
A state $\rho_{init}$ is a vertex of a set $\rho_{init}^{TP}$.
\end{rmk}

\begin{proof}
It trivially follows from Theorem \ref{Extremal} and the fact that for all states $\rho$, $\beta(\rho)$ is tightly thermomajorized by itself. 
\end{proof}

\begin{rmk}\label{interior}
A state $\rho\in\rho_{init}^{TP}$ lies on a face of $\rho_{init}^{TP}$ if and only if at least one of the elbows of $\beta(\rho)$ lies on $\beta(\rho_{init})$.
\end{rmk}

\begin{proof} [Proof: if part]
It is enough to show that for every $\rho\in\rho_{init}^{TP}$, with at least one of the elbows of $\beta(\rho)$ lying on $\beta(\rho_{init})$, there is a state $\rho_{1}\in\rho_{init}^{TP}$  such that there is no state  $\rho_{2}\in\rho_{init}^{TP}$ satisfying $\rho=\lambda\rho_{1}+(1-\lambda)\rho_{2}$, $\lambda\in[0,1]$. Choose $\rho_{1}$ such that $\beta(\rho_{1})$ has elbows on the same positions as $\beta(\rho)$, apart from the one lying on a vertical line that goes through a selected elbow of $\beta(\rho)$ lying on $\beta(\rho_{init})$: place
this elbow $\delta$-distance below the elbow $\beta(\rho)$. We can always choose $\delta>0$ to be small enough such that $\pi(\rho)=\pi(\rho_{1})$. In this case, in order for $\rho=\lambda\rho_{1}+(1-\lambda)\rho_{2}$ to be satisfied, all elbows of $\beta(\rho_{2})$ have to lye on elbows of $\beta(\rho_{1})$, apart from the one that lies on a vertical line that goes through a selected elbow of $\beta(\rho)$ lying on $\beta(\rho_{init})$: this elbow has to lye $\epsilon>0$ above the elbow of $\beta(\rho_{init})$. But this implies that $\rho_{2}\not\in\rho_{init}^{TP}$. 
\end{proof}

\begin{proof}[Proof: only if part]
If a state $\rho$ lies on a face of $\rho_{init}^{TP}$, then there exists some state $\rho_{1}\in\rho_{init}^{TP}$ such that there is no state $\rho_{2}\in\rho_{init}^{TP}$ which satisfies  $\rho=\lambda\rho_{1}+(1-\lambda)\rho_{2}$, $\lambda\in[0,1]$. Assume that all elbows of $\beta(\rho)$ lie below $\beta(\rho_{init})$. We will show that it leads to a contradiction, i.e. that for an arbitrary state $\rho_{1}\in\rho_{init}^{TP}$ we can construct a state $\rho_{2}$ which satisfies $\rho=\lambda\rho_{1}+(1-\lambda)\rho_{2}$. For curves that have all elbows below an initial curve, we can always modify the procedure of decomposing a state $\rho$ into extremal points of $\rho_{init}^{TP}$ (only if part of the proof of Theorem \ref{Extremal}) by taking $\gamma_{1,1}'$ such that $\pi(\gamma_{1,1}')=\pi(\beta(\rho_{1}))$ and $\gamma_{1,2}'$ such that $\pi(\gamma_{1,2}')=P_{d-1,d}\pi(\beta(\rho_{1}))$ in the first step, and then carry on with the procedure. At the end, we have $c_{d-1,1}=\lambda$, $\rho(\gamma_{d-1,1})=\rho_{1}$ and $\sum_{i=2,2^{d-1}}c_{d-1,i}=(1-\lambda)$,  $\sum_{i=2,2^{d-1}}\rho(\gamma_{d-1,i})=\rho_{2}$.   
\end{proof}

Therefore, we arrive with the following observation:
\begin{rmk}\label{interior}
Interior of the set $\rho_{init}^{TP}$ is composed by states $\rho$ such that $\beta(\rho)$ has all elbows below $\beta(\rho_{init})$.
\end{rmk}

Therefore, 
\begin{rmk}
For  $\rho_{init}\neq\rho_{\beta}$ (i.e., for $\beta(\rho_{init})$ with some slopes different), a Gibbs state $\rho_{\beta}=e^{-\beta H}/tr[e^{-\beta H}]$ lies in the interior of a set $\rho_{init}^{TP}$.
\end{rmk}

For a $d$-dimensional system, we can obtain a classification of hyperfaces of $\rho_{init}^{TP}$:

\begin{defn}[Hyperface]
A hyperface of a polytope $\rho_{init}^{TP}$ is a convex subset $H$ of states $\rho\in\rho_{init}^{TP}$ which cannot be expressed as a non-trivial convex combination of states from $\rho_{init}^{TP}/H$. 
\end{defn}

 For a given $\rho_{init}$, let us denote by $\mathbb{S}_{\rho_{init}}$ a set of all possible non-empty sets of elbows of curves tightly thermomajorized by $\beta(\rho_{init})$. For every $S\in\mathbb{S}_{\rho_{init}}$, we define $H_{S}$ as a set of all states $\rho$ such that their thermomajorization curves $\beta(\rho)$ coincide with $\beta(\rho_{init})$ exactly on $S$. Every element of the set $\mathbb{S}_{\rho_{init}}$ defines a hyperface:

\begin{thm}
$H_{S}$ is a hyperface of $\rho_{init}^{TP}$.
\end{thm}

\begin{proof}
Assume that $H_{S}$ is not a hyperface. Therefore, there is $\rho$ from $H_{S}$ which can be represented as a convex combination $\rho=\lambda\rho_{1}+(1-\lambda)\rho_{2}$, $\lambda\in[0,1]$, $\rho_{1,2}\in\rho_{init}^{TP}$,  such that at least one state, $\rho_{1}$, belongs to $\rho_{init}^{TP}/H_{S}$. 
It implies that, on a vertical line passing through some point of $\beta(\rho_{init})$, being an elbow of some curve tightly thermomajorized by $\beta(\rho_{init})$ and belonging to $S$, $\beta(\rho_{1})$ lies below this elbow (see Fig. \ref{proof}). 
Let us denote state populations by vectors $\v{p}$, $\v{r}$ and $\v{q}$: $\rho_{1}=\text{diag}[\v{p}]$, $\rho_{2}=\text{diag}[\v{r}]$, $\rho=\text{diag}[\v{q}]$.

\begin{figure}[h]
\centering
\includegraphics[width = 1\linewidth]{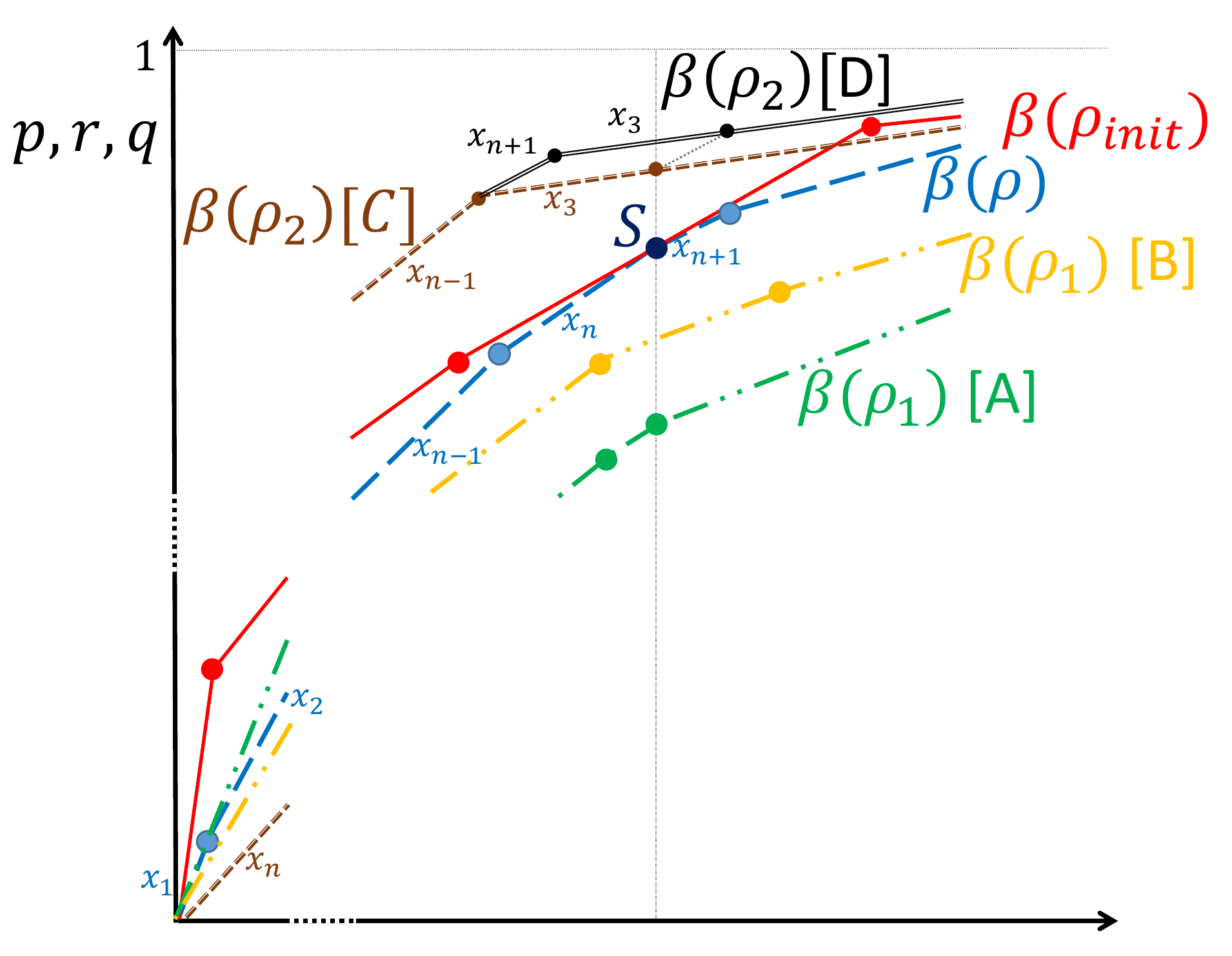}
\caption{\label{proof} 
A state $\rho$ belongs to a facet $H_{S}$ of $\rho_{init}^{TP}$ given by a set of points $S$. A point from this set defines a set of segments $\mathcal{X}=\{x_{1},\dots,x_{n}\}$ that lie to the left of it on $\beta(\rho)$. If $\rho_{1}\notin H_{S}$, then $\beta(\rho_{1})$ has an elbow below $S$ (case [A]) or a segment below $S$ (case [B]). Then, the state $\rho_{2}$ from the decomposition $\rho=\lambda\rho_{1}+(1-\lambda)\rho_{2}$, $\lambda\in[0,1]$ has a curve that has to lie over $S$, both in the case [C] with all segments of $\beta(\rho_{2})$ to the left from $S$ being taken from the set $\mathcal{X}$, as well as for different arrangements (case [D]). Therefore, $\rho_{2}\notin\rho_{init}^{TP}$. 
}
\end{figure}

Assume that $\beta(\rho_{1})$ has an elbow on this line (case A). If we denote by $\mathcal{X}=\{x_{1}, \dots, x_{n}\}$ a set of segments lying to the left of the elbow on $\beta(\rho)$, it is clear that, in order to have $\sum_{\mathcal{X}}q_{x}=\lambda \sum_{\mathcal{X}}p_{x}+(1-\lambda)\sum_{\mathcal{X}}r_{x}$, $\rho_{2}$ has to satisfy $\sum_{\mathcal{X}}q_{x}<\sum_{\mathcal{X}}r_{x}$ if elbow of $\beta(\rho_{1})$ lies below the elbow of in $S$ ($\sum_{\mathcal{X}}p_{x}<\sum_{\mathcal{X}}q_{x}$). This is because, if all segments of $\beta(\rho_{1})$ which lie to the left from the elbow belong to $\mathcal{X}$, we have $\sum_{\mathcal{X}}p_{x}<\sum_{\mathcal{X}}q_{x}$. If some segments from $\mathcal{X}$ lie to the right from the elbow on $\beta(\rho_{1})$, it means that $\sum_{\mathcal{X}}p_{x}$ is even smaller. Assume that segments from $\mathcal{X}$ of $\rho_{2}$ all lie to the left of the elbow (case $C$). Then, as $\sum_{\mathcal{X}}q_{x}<\sum_{\mathcal{X}}r_{x}$, we see that $\beta(\rho_{2})$ is not thermomajorized by $\beta(\rho_{init})$, and therefore $\rho_{2}\notin\rho_{init}^{TP}$. On the other hand, if some segments from the set $\mathcal{X}$ lie to the right of the elbow (case $D$), it implies that other segments in $\beta(\rho_{2})$ have even higher slopes, and the curve $\beta(\rho_{2})$ reaches even higher on a vertical line passing through the elbow, therefore $\rho_{2}\notin\rho_{init}^{TP}$. 

Now we will consider the case of $\beta(\rho_{1})$ not having an elbow on the vertical line passing through an elbow belonging to $S$ (case B). Again, it implies that $\sum_\mathcal{X}p_{x}<\sum_\mathcal{X}q_{x}$, as otherwise $\sum_\mathcal{X}q_{x}\leq\sum_\mathcal{X}p_{x}$ and the curve $\beta(\rho_{1})$ would not be thermomajorized by $\beta(\rho_{init})$ or $\beta(\rho_{1})$ would not lie below $S$ (as before, for $\beta(\rho_{1})$ to have some segments with higher slopes than segments from $\mathcal{X}$, it only increases the height of $\beta(\rho_{1})$ curve over the elbow from $S$ on $\beta(\rho_{init})$). Therefore, $\sum_\mathcal{X}p_{x}<\sum_\mathcal{X}q_{x}<\sum_\mathcal{X}r_{x}$ in order to have $\sum_{\mathcal{X}}q_{x}=\lambda \sum_{\mathcal{X}}p_{x}+(1-\lambda)\sum_{\mathcal{X}}r_{x}$, but this, as we showed before, leads to $\rho_{2}\notin\rho_{init}^{TP}$ (contradiction). 
\end{proof}

If we take $S_{1}, S_{2},\dots, S_{k}\in S_{\rho_{init}}$ such that $S_{1}\subseteq S_{2} \subseteq \dots \subseteq S_{k}$, then $H_{S_{1}}\supseteq H_{S_{2}} \supseteq \dots \supseteq H_{S_{k}}$. In particular, we see that every extremal point $\rho$ of $\rho_{init}^{TP}$ such that $\beta(\rho)$ has an elbow on $\beta(\rho_{init})$ on a given point, belongs to all facets of $\rho_{init}^{TP}$ which are composed by states with thermomajorization curves overlapping with $\rho_{init}$ on this point. Moreover, characterization of hyperfaces of the set of achievable states in terms of non-empty sets $S$ is complete:

\begin{rmk}
Every hyperface of $\rho_{init}^{TP}$ is $H_{S}$ for some $S\in S_{\rho_{init}}$.
\end{rmk}

If this was not true, then there would be some hyperface containing a state $\rho$ with $\beta(\rho)$ such that it has no elbows on $\rho_{init}$. From Remark \ref{interior} it stems that such a state belongs to the interior of $\rho_{init}^{TP}$, and therefore does not belong to any hyperface of $\rho_{init}^{TP}$.

Finally, we identify nearest neighbors of all extreme points by the following lemma:

\begin{lem}
For a state $\rho_{init}$ with all slopes of $\beta(\rho_{init})$ different, two distinct extremal states $\epsilon_{1}$, $\epsilon_{2}\in\rho_{init}^{TO}$, with orders of thermomajorization curves $\pi_{1}$ and $\pi_{2}$, respectively, are connected by an edge iff $\pi_{1}=P_{d-i,d-i+1}\pi_{2}$, for some $i\in\{2,\dots,d\}$.   
\end{lem}

\begin{proof}[Proof: if part]
For two extremal points to share an edge means that for all states $\rho_{\lambda}=\lambda\epsilon_{1}+(1-\lambda)\epsilon_{2}$, $\lambda\in(0,1)$, belonging to $\rho_{init}^{TP}$, and for every state $\sigma\in\rho_{init}^{TP}$ and $\sigma\notin\{\rho_{\lambda}\}_{\lambda}$, there is no state $\sigma'\in\rho_{init}^{TP}$ such that $\rho_{\lambda}=\gamma\sigma+(1-\gamma)\sigma'$ for $\gamma\in[0,1]$. For an arbitrary $\sigma\notin\{\rho_{\lambda}\}_{\lambda}$, there is at least one elbow of $\beta(\rho_{\lambda})$ lying on $\beta(\rho_{init})$, such that $\beta(\sigma)$ lies below it.

Assume that, looking from the left side of the thermomajorization diagram, the first elbow of $\beta(\rho_{\lambda})$ satisfies this property. It means that the segment of $\beta(\sigma)$ of the same length as the first segment of $\beta(\rho_{\lambda})$ has to have smaller slope than the slope of this segment in $\beta(\rho_{\lambda})$. Therefore, a corresponding segment in $\beta(\sigma')$ has to be have bigger slope than the slope of this segment in $\beta(\rho_{\lambda})$, as only in this way we can achieve $\rho_{\lambda}=\gamma\sigma+(1-\gamma)\sigma' \iff \forall_{i}\beta_{i}(\rho_{\lambda})=\gamma\beta_{i}(\sigma)+(1-\gamma)\beta_{i}(\sigma')$. But this leads to a contradiction with the requirement that $\sigma\in\rho_{init}^{TP}$, as $\beta(\sigma')$ would not be thermomajorized by $\beta(\rho_{init})$.  

If we assume that $\beta(\sigma)$ coincides with $\beta(\rho_{\lambda})$ on its first elbow lying on $\beta(\rho_{init})$, but the second such elbow of $\beta(\rho_{\lambda})$ lies above $\beta(\sigma)$, then  it means that $\beta(\sigma)$ had to have an elbow on the first elbow of $\beta(\rho_{\lambda})$ -- otherwise, $\beta(\rho_{\lambda})$ would not be thermomajorized by $\beta(\rho_{init})$ for all of its slopes different. Therefore, we conclude that $\beta(\rho_{\lambda})$, $\beta(\sigma)$ and $\beta(\sigma')$ are identical on their first segments. We can therefore treat the first elbow of $\beta(\rho_{\lambda})$ as the effective start of a new thermomajorization diagram, and apply the argument from the last step again.

We continue doing so for all the segments of $\beta(\rho_{\lambda})$, until we reach a segment $d-i$. If no elbow of $\beta(\rho_{\lambda})$ is lying above $\beta(\sigma)$ on this side of the original thermomajorization diagram, we apply the same reasoning the the right side of the diagram, until we reach the segment $d-i+1$. In this way, we are guaranteed to find an elbow of $\beta(\rho_{\lambda})$ that lies above $\beta(\sigma)$, as otherwise $\sigma\in\rho_{\lambda}$. At such event, we reach a conclusion $\sigma'\notin\rho_{init}^{TP}$, as shown above. 
\end{proof}

\begin{proof}[Proof: only if part]
Assume that for two distinctive extremal states $\epsilon_{1}$ and $\epsilon_{2}$, their respective orders cannot be related via $\pi_{1}\neq P_{d-i,d-i+1}\pi_{2}$, for any of $i\in\{2,\dots,d\}$. It means that a construction $\rho_{\lambda}=\lambda\epsilon_{1}+(1-\lambda)\epsilon_{2}$ for $\lambda\in(0,1)$ results in $\beta(\rho_{\lambda})$ that has at least 2 elbows below $\beta(\rho_{init})$, as $ P_{d-i,d-i+1}$ is the only relation between the orders of distinctive extremal states that leads to 1 elbow below  $\beta(\rho_{init})$ for all slopes of $\beta(\rho_{init})$ different. According to the procedure for decomposing a given state into extremal states (only if part of the proof of Theorem \ref{Extremal}), every elbow of $\rho_{\lambda}$ leads to a generation of 2 extremal points in the convex decomposition of this state into extremal points of $\rho_{init}^{TP}$. Therefore, $\rho_{\lambda}$ can be decomposed into at least 4 states, which contradicts the uniqueness of the decomposition $\rho_{\lambda}=\lambda\epsilon_{1}+(1-\lambda)\epsilon_{2}$, $\lambda\in(0,1)$.  
\end{proof}

\section{Geomery of the set of TPs}\label{sec:TPs}

In this section we are going to use properties of biplanar transportation matrices. Every Thermal Process $T$ acting on a $d$ level system can be turned into a transformation matrix $P$ by a transformation $P=T \text{diag}[1,q_{10}\dots,q_{d-1,0}]$. $P$ is characterized by vectors $\v{r}=\v{c}=[1,q_{10}\dots,q_{d-1,0}]$. $P$ and $T$ have identical adjacency matrices, and therefore the same graph representations. Therefore, all extremal TPs can be associated with forests with no isolated cycles on bipartite graphs, having at most $2d-1$ positive entries (note however, that some forests with isolated vertices may exist only for a specific choice of $\v{r}$ and $\v{c}$, so not all of them lead to extremal Thermal Processes). Moreover, extremal TPs that correspond to biplanar extremal transportation matrices play a special role in the characterization of transitions allowed by Thermal Operations. Every such TP can be attributed two quantities: an order $\pi_{in}(T)$, which is a sequence of labels on the left side of the bipartite graph of the associated transportation matrix $P$, and $\pi_{out}(T)$, which is a sequence of labels on the right side of the bipartite graph of the associated transportation matrix $P$, such that for these sequences the graph is plain. Note that these orders may not be given uniquely.  

\subsection{Biplanar Extremal Thermal Processes}

\begin{lem}\label{14}[Tight thermomajorization relation on states defines a biplanar extremal Thermal Process]
Every pair of states $\rho_{out}$, $\rho_{init}$, such that $\v{p}\coloneqq\beta(\rho_{init})$ tightly thermomajorizes $\v{r}\coloneqq\beta(\rho_{out})$, determines a biplanar extremal Thermal Process $T$ such that $T\rho_{init}=\rho_{out}$, $\pi_{in}(T)=\pi(p)$ and $\pi_{out}(T)=\pi(r)$. If all slopes of ${p}$ are different, then $T$ is the only TP that transforms $\rho_{in}$ into $\rho_{out}$. 
\end{lem}

\begin{figure}[h]
\centering
\includegraphics[width = 1\linewidth]{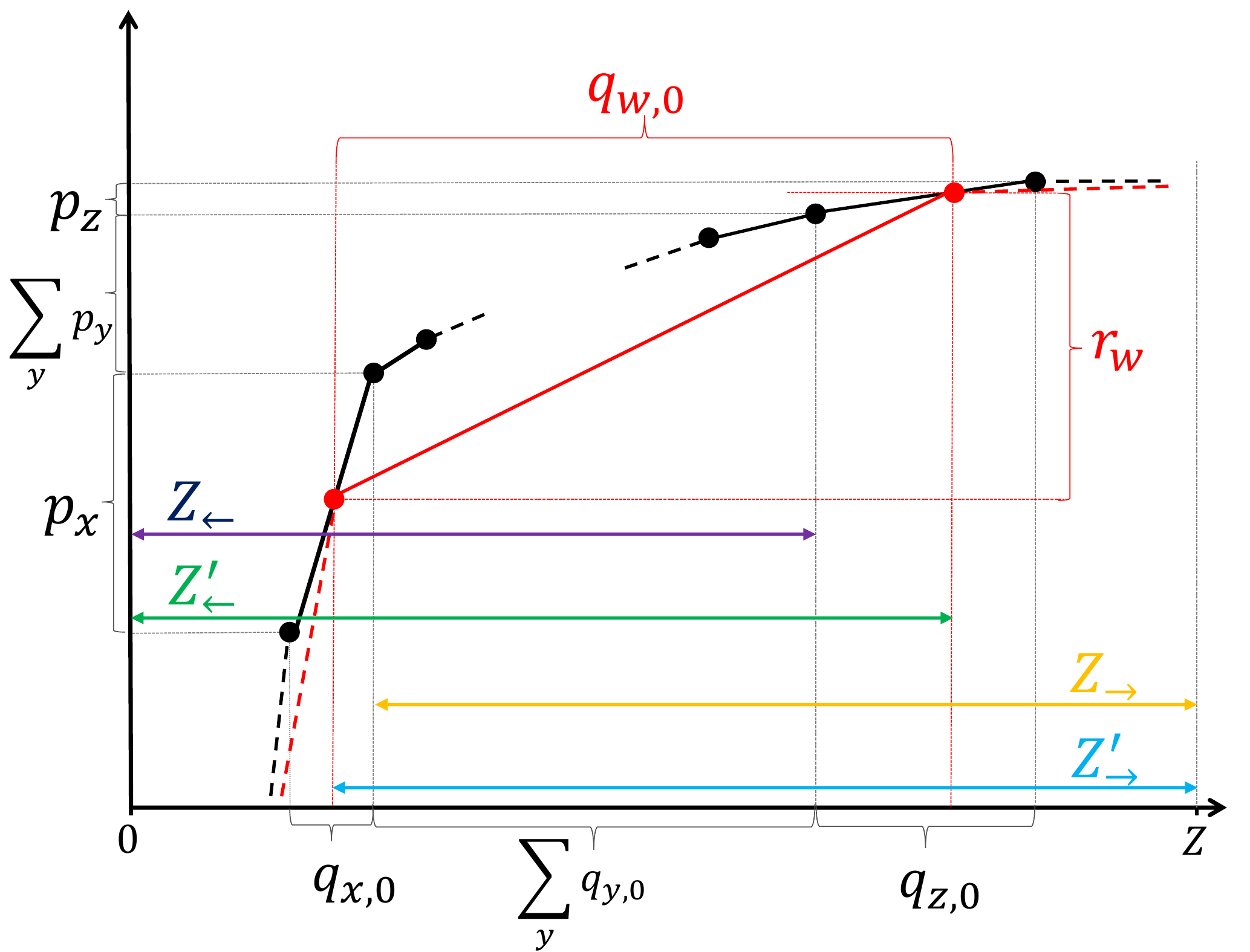}
\caption{\label{Gibbs} 
Construction of an extremal Thermal Process from thermomajorization diagram.
}
\end{figure}

\begin{proof}
$ $\newline

\underline{Thermal Process: Gibbs state preservation}.
Every $\beta$ fixes values of $q_{k,0}$ coefficients that determine lengths of segments of curves $\v{p}$ and $\v{r}$ on a thermomajorization diagram, and the association between states and TPs is done on the basis of thermomajorization curves $\v{p}$ and $\v{r}$.  If $\v{r}$ is tightly thermomajorized by $\v{p}$, then, from the thermomajorization diagram we propose a procedure that determines all components of the transformation. We will denote a slope of a segment $i$ on curve $\v{p}$ as $\partial\v{p}_{i}$, and, with a small abuse of notation, we will represent the state associated to the curve $\v{p}$ as a vector $\v{p}$, with entry $\v{p}_{i}$ on $i$-th level. 

In the case presented in Fig. \ref{Gibbs}, $\v{r}_{w}$ receives contributions from multiple segments of $\v{p}$. It is visible that for every $w$, $\v{r}_{w}$ can be formed from complete contributions from some levels of $\v{p}$ (we will label these levels by $y$), and at most two partial contributions from segments $x$ and $z$ of $\v{p}$ . These two partial contributions are a product of slopes $\partial\v{p}_{x}$ and $\partial\v{p}_{z}$ of respective segments of the curve $\v{p}$, and the lenghts of these segments. These lengths can be calculated as differences in components of partition function. We will denote them by $Z_{\rightarrow}=\sum {q_{i,0}/Z}$ for a sum over lengths of segments situated to the right from the point where the first segment included in $\sum \v{p}_y$ originates, while $Z_{\leftarrow}=\sum {q_{i,0}/Z}$ is a sum over lengths of segments situated to the left from the point which the last segment included in $\sum \v{p}_y$ has reached. $Z^{'}_{\rightarrow}$ and $Z^{'}_{\leftarrow}$ are defined analogously, but now with reference points changed to be initial and end points of the segment $w$ of the curve  $\v{r}$. From the definition, length of this segment is $q_{w,0}$, while its height is $\v{r}_{w}$. We arrive with the following formula describing the map transforming state $\v{p}$ into $\v{r}$ which is tightly thermomajorized by $\v{p}$: $\v{r}_{w}=\partial\v{p}_{x}(Z^{'}_{\rightarrow}-Z_{\rightarrow})  +\sum\limits_{y}\v{p}_{y} + \partial\v{p}_{z}(Z'_{\leftarrow}-Z_{\leftarrow})$. We have to show that this transformation is Gibbs preserving. In general we have $\partial\v{p}_{x}=p_{x}q_{0,x}$. If we start from a Gibbs state curve $\v{p}_{x}=q_{x,0}/Z$, then we arrive with $\v{r}_{w}=q_{x,0}q_{0,x}/Z(Z'_{\rightarrow}-Z_{\rightarrow})  +\sum\limits_{y}q_{y,0}/Z + q_{z,0}q_{0,z}/Z(Z'_{\leftarrow}-Z_{\leftarrow})=\Big(Z'_{\rightarrow}-Z_{\rightarrow}  +\sum\limits_{y}q_{y,0} + Z'_{\leftarrow}-Z_{\leftarrow} \Big)/Z$. But from Fig. \ref{Gibbs} it is visible that $-Z_{\rightarrow}-Z_{\leftarrow}+\sum\limits_{y}q_{y,0}=-Z$. Therefore, we have $\v{r}_{w}=\Big(Z'_{\rightarrow}+ Z'_{\leftarrow}-Z \Big)/Z=\Big(Z+q_{w,0}-Z \Big)/Z=q_{w,0}/Z$, which is a coefficient of a Gibbs state. 

The same conclusion holds for all other possible constructions of $\v{r}_{w}$: if $Z'_{\rightarrow}=Z_{\rightarrow}$, then $\sum\limits_{y}p_{y}=Z'_{\leftarrow}-Z_{\leftarrow}$, while for $Z'_{\leftarrow}=Z_{\leftarrow}$, we have $\sum\limits_{y}p_{y}=Z'_{\rightarrow}-Z_{\rightarrow}$, and one reaches the same conclusion about Gibbs state preservation of the process. Moreover, for all elbows of $\v{r}$ lying within one segment $f$ of $\beta(\v{p})$, we have  $Z_{\leftarrow}+Z_{\rightarrow}=Z$  and have $\sum\limits_{y}p_{y}=0$, for which $\v{r}_{w}=p_{f}q_{0f}q_{w0}$, which again equals to $q_{w,0}$ for initial Gibbs state $\v{p}_{f}=q_{f,0}/Z$. This exhausts the set of all possible geometrical relations between $\v{p}$ and a selected segment $\v{r}_{w}$ of the curve $\v{r}$.    

\underline{Thermal Process: Stochasticity}. Stochasticity of the transformation stems directly from the fact that $\v{r}$ is a curve on a thermomajorization diagram corresponding to a state: therefore every element $\v{p}_{i}$ is fully distributed into some set or $\{\v{r}_{j}\}_{j}$ (elements of every column of $T:T\rho_{init}=\rho_{out}$ sum to 1).

\underline{Extremality and biplanarity}. 
To show that $T$ is extremal and biplanar, it is enough to find a graph associated with $P$ which is plain, and to show that the graph is a forest with no isolated vertices. We will construct this graph by connecting a vertex $a$ from the right and $b$ from the left side of the graph whenever a ${b}$ on $\v{p}$ lies on the thermomajorization diagram above a segment $a$ on $\v{r}$, as it signifies the positive coefficient in $P$ on the position $(a,b)$ (and therefore a positive element in $T$). This leads to a graph of forests (because once a particular segment is considered, it does not reapper after we move to another segment on the same curve, so cycles are not possible) with no isolated vertices (because every segment lies below or above at least one segment). Furthermore, the graph is plain, with orders $\pi_{in}(T)$ and $\pi_{out}(T)$ fixed to be the same as orders $\pi(\v{p})$ and $\pi(\v{r})$, respectively. We see that this sequence of mappings is the one in which elements of the transportation matrix $P$ are fixed according to the defining procedure in Definition \ref{Def4}. 

\underline{Uniqueness}. Notice that the order of $\pi(\v{r})$ may not be given uniquely, as it is in principle possible to obtain a state with a curve $\v{r}$ that has more than one segments with the same slopes. On the other hand, if more than one segment has the same slope in $\v{p}$, it means that the transformation between the states may not be unique: In the extremal case, if all slopes of $\v{p}$ are the same, every TP performs the mapping (as every TP preserves a Gibbs state $\rho_{\beta}$). Therefore, we demand that all slopes in $\v{p}$ are different: it prohibits segments of $\v{p}$ to be permuted and fixes the sequence of points $(i,j)$ that the procedure in Definition \ref{Def4} utilizes to construct an associated transportation matrix. For every thermal process $T$ there is an associated matrix $T^{s}\coloneqq \frac{1}{Z^{2}}\text{diag}[1,q_{0,1},\dots,q_{0,d-1}] T \text{diag}[1,q_{1,0},\dots,q_{d-1,0}] $ that transforms slopes of the thermomajorization curve $\v{x}$ into slopes of $\v{y}$: $T\v{x}=\v{y} \iff T^{s}\partial\v{x}=\partial\v{y}$. $T^{s}$ has the same adjacency matrix as $T$, and therefore is associated with the same transportation matrix $P$ as $T$ is. Therefore, we see that the condition that $\v{r}$ is tightly thermomajorized by $\v{p}$ implies a map $T^{s}$ that is unique for all slopes of $\v{p}$ different, as only this map assures the curve $\v{r}$ has every of its elbows as high as possible (i.e. on $\v{p}$), given position of elbows to the left. Therefore, $T$ is also set uniquely. The property of the map of pushing the elbows as high as possible  is resembled by construction of the corresponding transportation matrix $P$ (Definition \ref{Def4}), where we assign the value $\min(r_{i},c_{j})$ to a given row or column -- this value is the biggest possible under constraints of $\v{r}$ and $\v{c}$.
\end{proof}

On the other hand, every biplanar extremal Thermal Process is associated with a pair of $\beta$-orders  of states, which are connected by a tight thermomajorization relation of their corresponding curves:

\begin{lem}\label{15}[every biplanar extremal Thermal Processes defines tight thermomajorization relation on states]
For an arbitrary biplanar extremal Thermal Process $T$ in the temperature $\beta$, characterized by orders $\pi_{in}(T)$ and $\pi_{out}(T)$, and a state $\rho_{init}$ such that $\pi(\beta(\rho_{init}))=\pi_{in}(T)$, we have that $\rho_{out}=T\rho_{init}\in Extr[\rho_{init}^{TP}]$ and $\pi(\beta(\rho_{out}))=\pi_{out}(T)$. If all slopes of $\beta(\rho_{init})$ are different, then $T$ is the unique transformation that maps $\rho_{in}$ into $\rho_{out}$. 
\end{lem}

\begin{proof}
Every biplanar extremal Thermal Process $T$ is characterized by sequences of $\pi_{in}(T)$ and $\pi_{out}(T)$, that label left and right side of the associated graph of a transportation matrix $P$, and for which the graph is plain. Therefore, for every initial state $\rho_{init}$ with order $\pi(\beta(\rho_{init}))=\pi_{in}(T)$, we obtain a state $\rho_{out}=T\rho_{init}$ with order $\pi(\beta(\rho_{out}))=\pi_{out}(T)$. Also, $\rho_{out}$ is an extremal point of $\rho_{init}^{TP}$, because $\beta(\rho_{out})$ provides the highest possible position for elbows for a given order $\pi(\beta(\rho_{out}))$, and the latter follows from the proof of uniqueness from Lemma \ref{14}. The proof for uniqueness of the transformation $T$ is the same as in Lemma \ref{14}. 
\end{proof}

To summarize, the relation between biplanar extremal TPs and states with tightly thermomajorizable relation can put as follows:

\begin{thm}\label{Tmain}
For a state $\rho_{init}$ with $\beta(\rho_{init})$ with all slopes different and $\beta$-order $\pi(\beta(\rho_{init}))$, $\rho_{out}$ with $\beta$-order $\pi(\beta(\rho_{out}))$ and a biplanar extremal Thermal Process $T$ with orders $\pi_{in}(T)$ and $\pi_{out}(T)$, respectively,
$\rho_{out}=T\rho_{init}\in Extr[\rho_{init}^{TP}]$ if and only if $\pi(\beta(\rho_{init}))=\pi_{in}(T)$ and $\pi(\beta(\rho_{out}))=\pi_{out}(T)$. $T$ is the only TP that satisfies $T\rho_{init}=\rho_{out}$.
\end{thm}

\begin{proof}
\textit{If part} follows directly from Lemma \ref{15}. \textit{Only if part} follows from Lemma \ref{14}.
\end{proof}

From Thereom \ref{Tmain} is follows that, for a selected temperature $\beta$, one can calculate all biplanar extremal Thermal Processes from thermomajorization diagrams by investigating all possible $\beta$-orders of initial and outoput states, where initial states have curves with all their slopes different, and they tightly thermomajorize curves of output states. Note that any change of $\beta$ influences the relations between different $q_{mn}$, which in turn influences possible $\beta$-orders of curves associated with $\rho_{out}\in Extr[\rho_{init}^{TO}]$. In this way, temperature-dependent geometry of $\rho_{init}^{TP}$ reflects temperature-dependent geometry of the set of Thermal Processes.

\subsection{Non-biplanar Extremal Thermal Processes}

We see that it is enough to be able to perform an arbitrary biplanar extremal TP, as it allows one to achieve an arbitrary extremal point of $\rho_{init}^{TP}$ for every $\rho_{init}$. Therefore, while extremal TPs that do not belong to a class of biplanar extremal TPs cannot be calculated from thermomajorization diagrams, they also seam to lack an operational meaning: when we allow for convex combinations of TPs, every state in the set of $\rho_{init}^{TP}$ can be achieved solely by the use of biplanar extremal TPs. Moreover, extremal TPs that are not biplanar cannot even lead to extremal points of $Extr[\rho_{init}^{TP}]$ for the case of different slopes of $\beta(\rho_{init})$; they always lead to the interior of the set, for every $\rho_{init}$. It stems from the uniqueness of $T$ for states with curves that have all slopes different. Naturally, in a degenerated case with some slopes in $\beta(\rho_{init})$ the same ($\rho_{init}=\rho_{\beta}$ being an extreme case), many processes may lead to the same state, so non-biplanar and biplanar extremal TPs can effectively coincide for this subset of possible $\rho_{init}$. 

A question arises if non-bipartite extremal TPs exist for given $d$. In fact, all extremal TPs for $d=2,3$ are biplanar (their list can be found in \cite{Mazurek2017}). For $d=4$, while we have shown a transportation matrix that is non-biplanar (Fig. \ref{forest}d), this construction is not valid for $\v{r}=\v{c}$, as in this case, a link connecting vertices labeled by 1 on both sides implies that there should be no link connecting a vertex `1' from the left with vertex `2' from the right. Therefore, we cannot construct a corresponding Thermal Process, for which it is necessary that $\v{r}=\v{c}=[1,q_{10}\dots,q_{d-1,0}]$. 

Extremal points of TPs have a very simple form for zero temperature. There, they have a vector $[1,0,\dots,0]^{T}$ as the first column, and independent permutations of this vector in different columns, e.g.  $\left(\begin{smallmatrix} 
1 & 1 & 0 & 0 \\
0 & 0 &  1 & 0\\
0 & 0 &  0 & 1 \\
0 & 0 & 0 & 0
\end{smallmatrix}\right)$. Therefore, they are transportation matrices. It is visible that these matrices are associated with graphs that are forests (each column has exactly one `1', so no loops are possible), but isolated vertices may be present (as there are some rows filled with `0'). Moreover, each of the rows of the matrix corresponds to an independent tree in a forest, so the graphs are biplanar. However, already for $d=4$, when going from zero to small temperatures, while all graphs become connected, some of them also become immediately non-biplanar. Consider the extremal TP:

\be
T_{non-biplanar}=
\begin{pmatrix}
1-q_{10} & 1 & 0 & 0\\
q_{10}-q_{20}  & 0 & 1  & 0 \\
q_{20}-q_{30} & 0 &  0 &  1\\
q_{30} &0 &  0 & 0\\
\end{pmatrix},
\ee
This process has an associated transportation matrix

\be
P(T_{non-biplanar})=
\begin{pmatrix}
1-q_{10} & q_{10} & 0 & 0\\
q_{10}-q_{20}  & 0 & q_{20}  & 0 \\
q_{20}-q_{30} & 0 &  0 &  q_{30}\\
q_{30} &0 &  0 & 0\\
\end{pmatrix},
\ee
described by a graph composed from a forest with no isolated vertices, shown in Fig. \ref{forest}e. 
Note that the construction can be trivially extended for arbitrary $d>4$. Therefore, non-bipartite extremal TPs are present for an arbitrary non-zero temperature for $d\geq 4$, and absent for $d=2,3$.\\

\section{Conclusions}
The established link between all physically significant extremal Thermal Processes and thermomajorization curves gives a recipe for determining the  form of relevant extremal TPs for systems of higher dimension. The complexity of the algorithm is the same as for determining of all the corresponding extremal transportation matrices. The number of extremal points of transportation polytopes is not known in general.   
 
With complete characterization of the set $\rho_{init}^{TP(d)}$ established, a similar description of $\rho_{init}^{TP(n)}$ for $n<d$ should allow for the solution of the decomposability problem $(n=2)$ by determining length of sequences of two-level transformations needed to penetrate $\rho_{init}^{TP(2)}$ for an arbitrary initial state.\\

\textbf{Acknowledgments.} 
The author would like to thank M. Horodecki for inspiring discussions. 
This work was supported by National Science Centre, Poland, grant OPUS 9. 2015/17/B/ST2/01945. 
 
\bibliographystyle{ieeetr}
\bibliography{Extremal_biblio2}	

\begin{thebibliography}{10}

\bibitem{Janzing00}
D.~Janzing, P.~Wocjan, R.~Zeier, R.~Geiss, and T.~Beth, ``Thermodynamic cost of
  reliability and low temperatures: Tightening landauer's principle and the
  second law,'' {\em International Journal of Theoretical Physics}, vol.~39,
  pp.~2717--2753, Dec 2000.

\bibitem{Streater95}
R.~Streater, {\em Statistical Dynamics: A Stochastic Approach to nonequilibrium
  Thermodynamics}.
\newblock Imperial College Press, London, UK, 1995.

\bibitem{Ruch76}
E.~Ruch and A.~Mead, ``The principle of increasing mixing character and some of
  its consequences,'' {\em Theoretica chimica acta}, vol.~41, pp.~95--117, Jun
  1976.

\bibitem{Horodecki13}
M.~{Horodecki} and J.~{Oppenheim}, ``{Fundamental limitations for quantum and
  nanoscale thermodynamics},'' {\em Nature Communications}, vol.~4, p.~2059,
  June 2013.

\bibitem{Aberg2013}
J.~{{\AA}berg}, ``{Truly work-like work extraction via a single-shot
  analysis},'' {\em Nature Communications}, vol.~4, p.~1925, June 2013.

\bibitem{Brandao2015}
F.~Brand{\~a}o, M.~Horodecki, N.~Ng, J.~Oppenheim, and S.~Wehner, ``The second
  laws of quantum thermodynamics,'' {\em Proceedings of the National Academy of
  Sciences}, vol.~112, no.~11, pp.~3275--3279, 2015.

\bibitem{Cwiklinski2015}
P.~\ifmmode \acute{C}\else \'{C}\fi{}wikli\ifmmode~\acute{n}\else
  \'{n}\fi{}ski, M.~Studzi\ifmmode~\acute{n}\else \'{n}\fi{}ski, M.~Horodecki,
  and J.~Oppenheim, ``Limitations on the evolution of quantum coherences:
  Towards fully quantum second laws of thermodynamics,'' {\em Phys. Rev.
  Lett.}, vol.~115, p.~210403, Nov 2015.

\bibitem{Alicki07}
R.~Alicki and K.~Lendi, {\em Quantum Dynamical Semigroups and Applications},
  vol.~717 of {\em Lecture Notes in Physics}.
\newblock Springer Berlin Heidelberg, 2007.

\bibitem{Perry16}
C.~Perry, P.~\'Cwikli\'nski, J.~Anders, M.~Horodecki, and J.~Oppenheim, ``A
  sufficient set of experimentally implementable thermal operations,'' {\em
  arXiv:1511.06553}, 2015.

\bibitem{Lostaglio16b}
M.~Lostaglio, {\'{A}}.~M. Alhambra, and C.~Perry, ``Elementary {T}hermal
  {O}perations,'' {\em {Quantum}}, vol.~2, p.~52, Feb. 2018.

\bibitem{Mazurek2017}
P.~{Mazurek} and M.~{Horodecki}, ``{Decomposability and convex structure of
  thermal processes},'' {\em New Journal of Physics}, vol.~20, p.~053040, May
  2018.

\bibitem{Jurkat1967}
W.~B. Jurkat and H.~J. Ryser, ``Term ranks and permanents of nonnegative
  matrices,'' {\em Journal of Algebra}, vol.~5, pp.~342--357, 1967.

\bibitem{Lostaglio15}
M.~Lostaglio, K.~Korzekwa, D.~Jennings, and T.~Rudolph, ``Quantum coherence,
  time-translation symmetry, and thermodynamics,'' {\em Phys. Rev. X}, vol.~5,
  p.~021001, Apr 2015.

\bibitem{Klee68}
V.~Klee and C.~Witzgall, {\em Facets and vertices of transportation polytopes}.
\newblock Providence, 1968.

\end{thebibliography}
\end{document}